\documentclass[a4paper,twocolumn,11pt,unpublished]{quantumarticle}
\pdfoutput=1
\usepackage[utf8]{inputenc}
\usepackage[english]{babel}
\usepackage[T1]{fontenc}
\usepackage{amsmath}
\usepackage{hyperref}
\usepackage[capitalise]{cleveref}
\usepackage{braket}
\usepackage{amssymb}
\usepackage{quantikz}
\usepackage{xcolor}
\usepackage{overpic}
\usepackage{booktabs}
\usepackage{amsthm}
\usepackage{amsmath}
\usepackage{bm}
\usepackage{comment}

\definecolor{softred}{HTML}{FF765F} 
\definecolor{strongred}{HTML}{B51700}
\definecolor{gold}{HTML}{DDC13F}
\definecolor{softblue}{HTML}{71BFFF}
\definecolor{strongblue}{HTML}{0076BA}
\definecolor{mygrey}{HTML}{5E5E5E}
\definecolor{quantum}{HTML}{53257F}

\hypersetup{
    colorlinks = true,
    linkcolor = quantum,
    citecolor = quantum
}

\newcommand{\norm}[1]{\left\lVert #1 \right\rVert}

\usepackage{tikz}
\usetikzlibrary{positioning}
\usepackage{lipsum}

\newtheorem{theorem}{Theorem}
\newtheorem{corollary}{Corollary}

\newtheorem{lemma}{Lemma}
\newtheorem{definition}{Definition}

\newcommand{\Tr}[1]{\text{Tr}\left( #1 \right)}

\let\vec\bm

\begin{document}
\title{Lower Bounds on Coherent State Rank}
\author{Florian Cottier}
\affiliation{DIENS, \'Ecole Normale Supérieure, PSL University, CNRS, INRIA, 45 rue d'Ulm, 75005 Paris, France}
\affiliation{\'Ecole Polytechnique Fédérale de Lausanne (EPFL), CH-1015 Lausanne, Switzerland}
\author{Ulysse Chabaud}
\affiliation{DIENS, \'Ecole Normale Supérieure, PSL University, CNRS, INRIA, 45 rue d'Ulm, 75005 Paris, France}

\begin{abstract} 
\onecolumn

The approximate coherent state rank is the minimal number of (classical) coherent states required to approximate a continuous-variable bosonic quantum state and directly relates to the classical complexity of simulating bosonic computations. Despite its importance, little is known about lower bounds on this quantity, even for basic families of states.
In this work, we initiate a systematic study of lower bounds on the approximate coherent state rank. Our contributions are as follows.

(i) We introduce a technique based on low-rank approximation theory yielding generic lower bounds on the approximate coherent state rank of arbitrary single-mode states.

(ii) Using this technique, we find a complete characterization of all single-mode states of finite approximate coherent state rank, and we obtain in particular analytical expressions for the approximate coherent state rank of squeezed states and of finite superpositions of Fock states.

(iii) We further show that our single-mode lower bounds can be lifted to multimode lower bounds for finite superpositions of multimode Fock states.

(iv) Finally, we prove a super-polynomial lower bound on the approximate coherent state rank of the $n$-mode Fock state $\ket{1}^{\otimes n}$, by exploiting a connection to the permanent. To do so, we show that the algebraic complexity of approximate multi-linear formulas for the permanent is super-polynomial, building upon the proof of a lower bound for exact formulas due to Raz \cite{raz2009multi}.

Our results establish an unconditional barrier to efficient classical simulation of Boson Sampling via coherent state decompositions and connect non-classicality of bosonic quantum systems to central questions in algebraic complexity.

\end{abstract}

\tableofcontents

\onecolumn
\section{Introduction}

The problem of identifying the features that fundamentally distinguish quantum systems from their classical counterparts dates back to the origins of quantum mechanics. In the modern language of quantum information theory, these distinctive features are formalized as resources enabling quantum advantage for information processing tasks \cite{chitambar2019quantum}. Characterizing, quantifying, and leveraging these resources remains a central challenge across quantum science \cite{preskill2018quantum,eisert2025mind}.

\smallskip

The main obstacle to harnessing quantum resources for information processing is that the mere presence of quantum effects does not guarantee a quantum advantage; indeed, some quantum systems can be efficiently simulated by classical computers. Prominently, \textit{coherent states}, describing for instance the quantum state of the light coming out of a laser, closely resemble classical electromagnetic waves and can be regarded as classical \cite{mandel1986non}. In contrast, non-classical states such as Schrödinger cat states (superpositions of coherent states) exhibit genuinely quantum features and play a key role in applications ranging from quantum sensing \cite{munro2002weak,joo2011quantum,facon2016sensitive} to quantum error correction \cite{mirrahimi2014dynamically,ofek2016extending,grimm2020stabilization}.

Historically, non-classicality in quantum physics has been characterized using phase-space concepts, whereby a state is deemed non-classical if its associated Glauber--Sudarshan phase-space distribution fails to behave as a genuine probability distribution \cite{glauber_1963,sudarshan_1963}.
This qualitative distinction was later refined into a quantitative hierarchy through the notion of \textit{degree of non-classicality} \cite{gherke_2012, Sperling_2015}. In this framework, coherent states are treated as the most classical states, and any pure state $\ket{\psi}$ is expressed as a superposition of coherent states. The degree of non-classicality is then defined as the minimal number of coherent states required in such a decomposition, namely the smallest integer $k$ such that $\ket{\psi} = \sum_{j=1}^k c_j \ket{\alpha_j}$ where $\ket{\alpha_j}$ are distinct coherent states and $c_j \in \mathbb{C}$, and thus ranks quantum states based on their non-classicality.

\smallskip

Beyond its fundamental relevance in quantum physics, the theory of non-classicality was leveraged more recently in the context of quantum computing with bosonic systems \cite{Marshall_2023}. These systems are ubiquitous in nature and have attracted growing attention for scalable, fault-tolerant quantum computing due to their remarkable quantum error correction properties \cite{Gottesman_2001,mirrahimi2014dynamically,michael2016new}. This potential has led to a series of experimental demonstrations of bosonic encodings across a wide range of architectures, including trapped ions, superconducting circuits, and photonic systems \cite{fluhmann2019encoding, konno2024, larsen2025integrated}.
In \cite{Marshall_2023}, Marshall and Anand establish a direct link between the non-classicality of a bosonic state in a computation, in the form of its so-called \textit{approximate coherent state rank}, and the classical computational overhead required for its simulation. The approximate coherent state rank generalizes the degree of non-classicality and is defined as the minimal size of superpositions of coherent states necessary to approximate a quantum state to arbitrary precision in fidelity. 

The relevance of this non-classicality measure stems from the fact that coherent states evolving under so-called Gaussian bosonic circuits can be efficiently simulated classically \cite{Bartlett_2002}. Supplemented with appropriate input states, Gaussian bosonic circuits implement universal quantum computations \cite{baragiola2019all}. As a consequence, the complexity of simulating a bosonic quantum computation is directly related to the approximate coherent state rank of the input states involved in the computation.
This connection raises the question of how large the approximate coherent state rank must be for physically relevant states, and whether it can serve as evidence of the classical intractability of bosonic computations.

\smallskip

This question mirrors a related line of research on the \textit{stabilizer rank} in qubit systems. For qubit computations, the Gottesman--Knill theorem \cite{Gottesman_1998} establishes that a quantum circuit taking computational basis states as input and performing only unitary gates from the Clifford group and Pauli measurements can be simulated efficiently by a classical computer, because the states throughout the computation are limited to \textit{stabilizer states}, which admit an efficient classical description. On the other hand, universal quantum computing can be achieved by supplementing such circuits with superpositions of stabilizer states via magic state injection \cite{Bravyi_2005}. By linearity, the size of these superpositions (their stabilizer rank) directly relates to the running time of classical simulation algorithms for the corresponding quantum computation. In particular, an efficient polynomial-size stabilizer rank decomposition of the magic state $\ket H^{\otimes n}$ would directly imply that classical computers can simulate their quantum counterparts, i.e.\ $\mathsf{BPP}=\mathsf{BQP}$. While separating the complexity classes $\mathsf{BPP}$ and $\mathsf{BQP}$ is a daunting challenge, researchers have focused on refining upper and lower bounds on the stabilizer rank of $\ket H^{\otimes n}$. To date, the best upper bounds are exponential both in the exact \cite{Bravyi_2016,kocia2020improved,qassim2021improved} and approximate case \cite{bravyi2016improved}, while lower bounds have also attracted significant attention \cite{Bravyi_2016,Peleg_2022,Lovitz2022newtechniques,mehraban2024}, the best known lower bounds being quadratic (up to logarithmic factors) for both the exact and approximate stabilizer rank \cite{mehraban2024}. Proving a super-polynomial lower bound on the stabilizer rank of $\ket H^{\otimes n}$ is an outstanding open question.

\smallskip

Akin to the stabilizer rank, the approximate coherent state rank has strong ties with complexity theory. Indeed, it was noticed by Troyansky and Tishby \cite{troyansky1996quantum} and Scheel \cite{scheel2004} that quantum amplitudes of bosonic systems are closely related to matrix permanents, a combinatorial quantity that is $\#\mathsf P$-hard to compute \cite{Valiant_1979}. Specifically, Scheel showed that, when $n$ identical bosons in the Fock state $\ket1^{\otimes n}$ are sent through a linear network, the probability amplitude of a specific output configuration is proportional to the permanent of a matrix describing the network.
A consequence of this result is that an efficient polynomial-size approximate coherent state rank decomposition of the bosonic state $\ket1^{\otimes n}$ would lead to $\mathsf P=\mathsf P^{\#\mathsf P}$ and in particular to a collapse of the whole polynomial hierarchy by Toda's theorem \cite{toda1991pp}, i.e.\ $\mathsf P=\mathsf{NP}$\footnote{Aaronson and Arkhipov famously exploited the $\#\mathsf P$-hardness of the permanent to prove the classical hardness of the sampling version of this problem, known as Boson Sampling \cite{aaronson_2010}, unless the polynomial hierarchy collapses to its third level.}. The permanent is also a central object in the study of algebraic complexity, being $\mathsf{VNP}$-complete \cite{valiant1979completeness}. The smallest known formulas for the permanent are of exponential size and attributed to Glynn \cite{glynn_2010} and Ryser \cite{Ryser_1991}, leading to the fastest known algorithms for computing the permanent. Recently, it was noticed in \cite{Chabaud_2022_permanent} that Glynn’s formula for the permanent can be obtained via an approximate coherent state decomposition of size $2^n$ of the state $\ket{1}^{\otimes n}$. As a result, showing that the approximate coherent state rank of $\ket{1}^{\otimes n}$ is less than $2^n$ would likely yield a faster algorithm to compute the permanent than via Glynn's (or Ryser’s) formula, while a polynomial approximate coherent state rank is unlikely unless $\mathsf{VP}=\mathsf{VNP}$.

\smallskip

These connections motivate the importance of bounding the approximate coherent state rank of bosonic quantum states: on the one hand, upper bounds lead to improved classical algorithms for the permanent and for simulating bosonic quantum computations; on the other hand, lower bounds provide quantitative certificates of non-classicality and unconditional results towards separating existing complexity classes. While some upper bounds on the approximate coherent state rank are available for specific states \cite{Marshall_2023,upreti2025bounding}, little is known about lower bounds.

\smallskip

In this work, we initiate the search for lower bounds on the approximate coherent state rank of bosonic quantum states. 
We first consider single-mode states and prove a generic lower bound on the approximate coherent state rank using low-rank approximation theory, which leads to tight lower bounds for specific families of bosonic quantum states. In particular, we find a complete characterization of single-mode states of finite approximate coherent state rank, leading to analytical expressions for the approximate coherent state rank of single-mode Fock states, finite superpositions of Fock states, and squeezed states. 
Next, we turn to multimode states and show that our single-mode results can be lifted to this setting, obtaining a non-trivial lower bound on the approximate coherent state rank of finite superpositions of multimode Fock states. 
Finally, we prove a super-polynomial lower bound on the approximate coherent state rank of the Fock state $\ket1^{\otimes n}$, by leveraging its connection with the algebraic complexity of the permanent.

\subsection{Technical overview of results}

In this section, we provide a technical overview of our results, including proof techniques.

\subsubsection{Definitions}

Let us start by introducing a few preliminary definitions. We refer the reader to \cref{sec:prelim} for further preliminary material and background.

\smallskip

Bosonic quantum states belong to an infinite-dimensional Hilbert space equipped with a countable basis. For single-mode systems, such an orthonormal basis is the so-called Fock basis denoted $\{\ket n\}_{n\in\mathbb N}$, while for $m$-mode systems ($m$ computational registers), the Fock basis is given by $\{\ket{\bm n}\}_{\bm n=(n_1,\dots,n_m)\in\mathbb N^m}$ where $\ket{\bm n}=\ket{n_1}\otimes\dots\otimes\ket{n_m}$. Canonical bosonic operators known as creation and annihilation operators, respectively, are defined by their action on the Fock basis as $\hat a^\dag\ket n=\sqrt{n+1}\ket{n+1}$ and $\hat a\ket n=\sqrt n\ket{n-1}$, with $\hat a\ket0=0$.

We introduce a few standard bosonic quantum states used in this work. \textit{Fock states} are elements of the Fock basis, with $\ket0$ being the \textit{vacuum state}, while finite superpositions of Fock states are also known as \textit{core states}.
Single-mode \textit{coherent states} are defined in the Fock basis as
\begin{equation}
    \ket{\alpha}=e^{-\frac12|\alpha|^2}\sum_{n=0}^{+\infty}\frac{\alpha^n}{\sqrt{n!}}\ket n,
\end{equation}
for $\alpha\in\mathbb C$, while $m$-mode coherent states are obtained via $\ket{\bm\alpha}=\ket{\alpha_1}\otimes\dots\otimes\ket{\alpha_m}$, for $\bm\alpha\in\mathbb C^m$. Single-mode coherent states are obtained from the vacuum state as $\ket\alpha=\hat D(\alpha)\ket0$, where $\hat D(\alpha)=e^{\alpha\hat a^\dag-\alpha^*\hat a}$ is the so-called displacement operator of amplitude $\alpha$.
Single-mode \textit{squeezed states} are defined in the Fock basis as
\begin{equation}
    \ket{\xi}=\frac1{\sqrt{\cosh(r)}}\sum_{n=0}^\infty\left(-e^{i\phi}\tanh(r)\right)^n\frac{\sqrt{(2n)!}}{2^nn!}\ket{2n},
\end{equation}
for $\xi=re^{i\phi}\in\mathbb C$.

\smallskip

The main focus of this paper, the $\varepsilon$-approximate coherent state rank, is defined as follows:
 
\begin{definition}[$\varepsilon$-approximate coherent state rank \cite{Marshall_2023,upreti2025bounding}]
Let $\varepsilon>0$. The $\varepsilon$-approximate coherent state rank of an $m$-mode state $\ket{\bm\psi}$, denoted $\kappa_\varepsilon({\bm\psi})$, is the smallest $k$ such that there exists a state $\ket{\bm\phi}=\sum_{j=1}^kc_j\ket{\bm\alpha_j}$, where $c_j\in\mathbb C$ and $\ket{\bm\alpha_j}$ are $m$-mode coherent states, such that $F(\bm\psi,\bm\phi)\geq1-\varepsilon$.
\end{definition}
 
\noindent Here $F$ denotes the fidelity. Note that the $\varepsilon$-approximate coherent state rank is either an integer or infinite.
In the limit where $\varepsilon$ goes to $0$, we refer to the $\varepsilon$-approximate coherent state rank simply as the approximate coherent state rank, defined by $\kappa=\sup_{\varepsilon>0}\kappa_\varepsilon=\lim_{\varepsilon\to0^+}\kappa_\varepsilon$.

\smallskip

With these preliminary definitions in place, we now turn to our results.

\subsubsection{Single-mode results}

Our first result is a generic lower bound on the $\varepsilon$-approximate coherent state rank of an arbitrary single-mode state $\ket{\psi}$, based on low-rank approximation theory. The result is stated informally hereafter, and we refer to \cref{thm:lower_bound_appox_coherent_stellar} for a formal statement and proof.

\begin{theorem}[Generic lower bounds on the $\varepsilon$-approximate coherent state rank (informal)]
Let $\ket{\psi}$ be a single-mode quantum state. There is a family of non-negative, efficiently computable functions $(f_N)_{N\in\mathbb N}$ such that
\begin{equation}
    \varepsilon < f_N(\psi,r)\quad\Rightarrow\quad\kappa_\varepsilon(\psi)>r,
\end{equation}
for all $r\le N\in\mathbb N$.
\label{thm:lower_bound_csr_informal} 
\end{theorem} 

\begin{proof}[Proof sketch]
    To prove the result, our main insight is that, up to a rescaling factor, the Fock basis amplitudes $(\sqrt{n!}\,\phi_n)_{n\in\mathbb N}$ of a superposition of $r$ coherent states $\ket\phi =\sum_{j=1}^r c_j \ket{\alpha_j}$ follow a linear recurrence relation of order $r$. This implies that the Hankel matrix $H_N(\phi)$ with entries $(\sqrt{(i+j)!}\,\phi_{i+j})_{0\le i,j\le N}$ associated to these amplitudes has rank at most $r$ \cite{gantmacher1959theory}, for any truncation parameter $N$.

    On the other hand, for a generic state $\ket\psi=\sum_{n=0}^{+\infty}\psi_n\ket n$, the Hankel matrix $H_N(\psi)$ with entries $(\sqrt{(i+j)!}\,\psi_{i+j})_{0\le i,j\le N}$ associated to its rescaled Fock basis amplitudes typically has higher rank. By the Young--Eckart--Mirsky theorem \cite{Eckart_Young_1936, mirsky1960symmetric} (see \cref{thm:YEM}), a fundamental result in low-rank approximation theory, this implies a minimal distance between the two matrices $H_N(\phi)$ and $H_N(\psi)$ in Frobenius norm, which only depends on the singular values of $H_N(\psi)$. We conclude by relating the distance between the Hankel matrices $H_N(\phi)$ and $H_N(\psi)$ to the fidelity between the quantum states $\ket\phi$ and $\ket\psi$. As a consequence of the proof, the function $f_N(\psi,r)$ in \cref{thm:lower_bound_csr_informal} can be directly computed from the singular values of $H_N(\psi)$.
\end{proof}

\cref{thm:lower_bound_csr_informal} provides a lower bound on the $\varepsilon$-approximate coherent state rank of any state $\ket{\psi}$ for each $f_N$, which can be refined by optimizing over the choice of the free parameter $N$. Note that the functions $f_N(\psi,r)$ can be zero for specific choices of parameters $N,r$ and for some states $\ket{\psi}$ (for instance when $r\ge\kappa(\psi)$), yielding, in these cases, a trivial statement. Importantly, however, \cref{thm:lower_bound_csr_informal} does provide powerful lower bounds on the $\varepsilon$-approximate coherent state rank for most choices of states $\ket{\psi}$, as we illustrate in \cref{app:optib}. 

\smallskip

We also obtain the following remarkable consequence of \cref{thm:lower_bound_csr_informal}, which provides a complete characterization of states of finite approximate coherent state rank (see \cref{thm:LLR} for a detailed proof):

\begin{theorem}[Characterization of single-mode states of finite approximate coherent state rank]
Let $\ket\psi=\sum_{n=0}^{+\infty}\psi_n\ket n$ be a single-mode quantum state and let $r>0$. Then, $\kappa(\psi)=r$ if and only if the sequence $(\sqrt{n!}\,\psi_n)_{n\in\mathbb N}$ follows a linear recurrence relation of order $r$; equivalently, if and only if $\ket\psi=\sum_{j=1}^kc_j\hat D(\alpha_j)\ket{C_j}$ is a superposition of $k\le r$ displaced core states with distinct $\alpha_j$ and where the core states have highest Fock number $n_j$ such that $\sum_{j=1}^k(n_j+1)=r$.
\label{thm:LLR1}
\end{theorem}

\begin{proof}[Proof sketch]
    To prove the forward implication, we show using \cref{thm:lower_bound_csr_informal} that $\kappa(\psi)=r$ implies that the Hankel matrix $H_N(\psi)$ with entries $(\sqrt{(i+j)!}\,\psi_{i+j})_{0\le i,j\le N}$ has rank at most $r$, for all $N\in\mathbb N$. As a consequence, the sequence $(\sqrt{n!}\,\psi_n)_{n\in\mathbb N}$ follows a linear recurrence relation of order $r$ \cite{salem1963algebraic} and $\sqrt{n!}\psi_n$ can be expressed as a finite sum of terms of the form $P(n)\alpha^n$, where $P$ is a polynomial and $\alpha\in\mathbb C$, for all $n\in\mathbb N$ \cite{rosen1999discrete}. From this expression, we deduce that $\ket\psi$ can be expressed as a finite superposition of displaced core states as in the theorem.

    To prove the reverse implication, we use the upper bound on the approximate coherent state rank of core states from \cite[Theorem 1]{Marshall_2023} to show that $\ket\psi=\sum_{j=1}^kc_j\hat D(\alpha_j)\ket{C_j}$ has approximate coherent state rank at most $r=\sum_{j=1}^k(n_j+1)$, where $n_j$ is the highest Fock number of the core state $\ket{C_j}$. Finally, we prove using a growth argument that the approximate coherent state rank cannot be smaller than $r$ when the displacement amplitudes $\alpha_j$ are distinct.
\end{proof}

\smallskip

We now present notable state-specific consequences of \cref{thm:LLR1}.
We first consider core states:

\begin{corollary}[Approximate coherent state rank of single-mode core states]
   Let $\psi$ be a finite superposition of Fock states with highest Fock number $n$. Then, its approximate coherent state rank is given by $\kappa(\psi)=n+1$.
   \label{cor:fock_lowerbound_csr}
\end{corollary} 

\begin{proof}
The result follows from a direct application of \cref{thm:LLR1} for $k=1$, $\alpha_j=0$ and $\ket{C_j}=\ket\psi$, with $n_j=n$.
\end{proof}

\noindent In particular, \cref{cor:fock_lowerbound_csr} implies that the $n^\text{th}$ Fock state requires a superposition of exactly $n+1$ coherent states to be approximated up to arbitrary precision, i.e.\ $\kappa(\ket n)=n+1$. We further obtain an analytical bound on how small $\varepsilon$ must be to ensure $\kappa_\varepsilon(\ket n)=n+1$ (see \cref{coro:coreFockACSR}).

\smallskip

Next, we prove that the approximate coherent state rank of a squeezed state is infinite: 

\begin{corollary}[Approximate coherent state rank of single-mode squeezed states]
    The $\varepsilon$-approximate coherent state rank of a squeezed state $\kappa_\varepsilon(\ket{\xi})$ with $\xi\neq0$ goes to infinity as $\varepsilon$ goes to $0$. Equivalently, \mbox{$\kappa(\ket{\xi})=+\infty$} for all $\xi\neq0$.
    \label{coro:no_finite_decomp_squeezed}
\end{corollary}

\begin{proof}[Proof sketch]
    To prove the result, we give a combinatorial proof that the sequence of rescaled Fock basis amplitudes $(\sqrt{n!}\,\xi_n)_{n\in\mathbb N}$ does not satisfy any finite linear recurrence relation with constant coefficients (see \cref{lem:sn_no_LRR}), and we conclude using \cref{thm:LLR1}.
\end{proof}

\subsubsection{Multimode results}

We now consider multimode quantum states, and we show that our single-mode lower bounds can be lifted to multimode ones for core states: 

\begin{theorem}[Lower bound on the approximate coherent state rank of multimode core states]\label{thm:multimode_core_state_lowerbound_csr}
   Let $\ket{\bm\psi}$ be a finite superposition of $m$-mode Fock states with total number of bosons at most $n$. That is \begin{equation}
    \ket{\bm\psi} = \sum_{|\bm{n}| \leq n} c_{\bm n} \ket{\bm{n}}, \quad \text{with } |\bm{n}| =  \sum_{i=1}^{m} \bm n_i,
\end{equation} 
where there exists at least one $\bm {n}$ with $|\bm{n}| = n$ such that $c_{\bm{n}} \neq 0$. Then,
\begin{equation}
    \kappa(\bm\psi)\ge n+1.
\end{equation}
\end{theorem}

\begin{proof}[Proof sketch]
The main idea behind the proof of the result is to show that any $m$-mode approximate coherent state decomposition for the state $\ket{\bm\psi}$ can be compressed to a single-mode approximate coherent state decomposition of the same size for a single-mode superposition of Fock states with highest Fock number $n$, in order to use the lower bound from \cref{cor:fock_lowerbound_csr}. 

We sketch the proof in the case where $\ket{\bm\psi}=\ket1^{\otimes n}$, the general case being analogous. We consider the following procedure: apply a specific passive linear unitary operation (see \cref{sec:linearoptics}), which does not change the total number of bosons, maps coherent states to coherent states, and compresses all the bosons in the first mode with a given non-zero probability; then, post-select on having no bosons in the other modes. Under this procedure, the state $\ket1^{\otimes n}$ is projected onto the single-mode Fock state $\ket n$, while any $n$-mode superposition of coherent states of size $r$ is mapped to a single-mode superposition of coherent states of size $r$. Finally, we carefully bound the evolution of the fidelity under this procedure, to show that an approximate coherent state decomposition for the state $\ket1^{\otimes n}$ leads to another approximate coherent state decomposition of the same size for the single-mode Fock state $\ket n$. We conclude using \cref{cor:fock_lowerbound_csr}.
\end{proof}

This result directly implies a linear lower bound on the approximate coherent state rank of the state $\ket1^{\otimes n}$, i.e.\ $\kappa(\ket1^{\otimes n})\ge n+1$.
The following result significantly improves the lower bound for the specific case of $\ket1^{\otimes n}$ by exploiting the connection between this state and the permanent (see \cref{thm:superpoly} for a detailed proof):

\begin{theorem}[Super-polynomial lower bound on the approximate coherent state rank of $\ket1^{\otimes n}$]\label{thm:tensor_fock_lowerbound_csr}
Let $\ket1^{\otimes n}$ denote the tensor product of $n$ Fock states $\ket{1}$. Then,
\begin{equation}
    \kappa(\ket1^{\otimes n})=n^{\Omega(\log n)}.
\end{equation}
\end{theorem}

\begin{proof}[Proof sketch]
    The proof strategy is as follows. We exploit the connection between the permanent and the state $\ket1^{\otimes n}$ in the context of bosonic computations. Namely, given an $n$-mode linear network $\hat U$ described by an $n\times n$ unitary matrix $U$, we have $\mathrm{Per}(U)=\bra1^{\otimes n}\hat U\ket1^{\otimes n}$ \cite{scheel2004,aaronson_2010}.
    
    First, we show that any $\varepsilon$-approximate coherent state decomposition of size $r$ for the state $\ket1^{\otimes n}$ leads to a polynomial formula $F_\varepsilon$ of size $rn^2$ which is multi-linear (that is, the power of every input variable is at most one in each of its monomials), and such that $F_\varepsilon$ uniformly approximates the permanent when evaluated on the entries of a unitary matrix $U$, namely \ $|F_\varepsilon(U)-\mathrm{Per}(U)|\le\sqrt{2\varepsilon}$. Letting $\varepsilon$ go to zero, this gives a sequence of multi-linear formulas of size $rn^2$ converging uniformly over the unitary group to the permanent. 
    
    Next, we prove that the uniform convergence over the unitary group implies a convergence coefficient-by-coefficient. In algebraic complexity terms, this implies that the border complexity of multi-linear formulas for the permanent is upper bounded by $rn^2$.
    
    Finally, we conclude by proving a super-polynomial ($n^{\Omega(\log n)}$) lower bound on the border complexity of multi-linear formulas for the permanent (see \cref{thm:bordercompperinformal} below).
\end{proof}

Crucial to the proof of \cref{thm:tensor_fock_lowerbound_csr} is the following algebraic complexity result, which extends a result due to Raz \cite{raz2009multi} on the size of multi-linear formulas for the permanent to approximate formulas (see \cref{thm:bordercompper} for a formal statement and proof):

\begin{theorem}[Border complexity of multi-linear formulas for the permanent (informal)]\label{thm:bordercompperinformal}
    Any sequence of multi-linear formulas converging to the permanent must contain only formulas of super-polynomial size beyond a certain index.
\end{theorem}

\begin{proof}[Proof sketch]
    The proof that multi-linear formulas for the permanent have super-polynomial size in \cite{raz2009multi} is based on a randomized construction of a so-called partial-derivatives matrix (PDM) from a random selection of the coefficients of a polynomial formula. In particular, Raz shows that the PDM built from the permanent polynomial is full-rank, while the PDM built from a multi-linear formula of size $n^{O(\log n)}$ has small rank with probability $1-o(1)$. Taken together, these statements imply a super-polynomial lower bound on the size of multi-linear formulas for the permanent.

    To prove the result, we follow a similar strategy. We first show that given a sequence of multi-linear formulas $F_\varepsilon$ converging to the permanent, the corresponding sequence of PDMs converges to the PDM of the permanent. Since the rank is lower semi-continuous, this implies that the PDMs built from $F_\varepsilon$ must have full rank for $\varepsilon$ small enough. Assuming that these multi-linear formulas $F_\varepsilon$ have size $n^{O(\log n)}$, we can use the results of Raz to argue that the PDMs built from $F_\varepsilon$ have small rank, each with probability $1-o(1)$. However, to conclude we need to show that with high probability, an infinite sequence of these PDMs has a small rank. To do so, we refine the argument and show that one can extract a subsequence $F_{\varphi(\varepsilon)}$ such that with probability $1-o(1)$, all the PDMs built from $F_{\varphi(\varepsilon)}$ have small rank. Taken together, these statements imply that any sequence of multi-linear formulas $F_\varepsilon$ converging to the permanent must contain only formulas of super-polynomial size for $\varepsilon$ small enough.
\end{proof}

\subsection{Discussion and open questions}

While some early numerical studies \cite{kenfack2004optimal} and witnesses \cite{mraz2014witnessing,kuhn2018quantum} are available, our work provides (to our knowledge) the first non-trivial lower bounds on the approximate coherent state rank. Regarding upper bounds, it is well-known that the Fock state $\ket1$ can be obtained as the limit of a superposition of two coherent states (an odd cat state of vanishing amplitude), i.e.\ $\kappa(\ket 1)\le2$. More generally, it was shown in \cite{janszky1993coherent,vogel_2014} that $\kappa(\ket n)\le n+1$. This bound was later extended in \cite{Marshall_2023} to all core states with highest Fock number at most $n$. \cref{cor:fock_lowerbound_csr} gives matching lower bounds and resolves the question of the approximate coherent state rank of single-mode Fock states and core states. 

In the context of the resource theory of non-classicality (see \cite{yadin2018operational,regula2021operational,lami2021framework} and references therein), \cref{cor:fock_lowerbound_csr} also demonstrates that Fock states form a strict hierarchy, i.e.\ the non-classicality of $\ket m$ is higher than that of $\ket n$ when $m>n$, as measured by the approximate coherent state rank. Moreover, \cref{coro:no_finite_decomp_squeezed} implies that squeezed states $\ket\xi$ sit above this infinite hierarchy of non-classicality for all $\xi\neq0$.
More generally, \cref{thm:LLR1} provides a complete characterization of single-mode quantum states of finite approximate coherent state rank as finite superpositions of displaced Fock states. This shows that most quantum states have infinite approximate coherent state rank and motivates the analysis of the more fine-grained $\varepsilon$-approximate coherent state rank.

By truncating a state $\ket\psi$ of bounded mean Fock number to a core state, the upper bound on the approximate coherent state rank of core states from \cite[Theorem 1]{Marshall_2023} directly yields an upper bound on the $\varepsilon$-approximate coherent state rank of $\ket\psi$. Using this technique, an upper bound on the $\varepsilon$-approximate coherent state rank of specific bosonic quantum states was obtained in \cite[Lemma 2]{upreti2025bounding}. In contrast, \cref{thm:lower_bound_csr_informal} provides computable lower bounds on the $\varepsilon$-approximate coherent state rank of arbitrary single-mode quantum states. It would be interesting to investigate how tight these bounds are; we take a first step in this direction in \cref{app:optib}.

Beyond non-classicality, other measures of quantumness have been introduced in the context of bosonic quantum computing \cite{kenfack2004negativity,mari2012positive,Chabaud_2020,Chabaud2023,Dias_2024,hahn2025classicalsimulationquantumresource}. In particular, the (approximate) Gaussian rank \cite{Dias_2024,hahn2025classicalsimulationquantumresource} is defined as the minimal size of a superposition of Gaussian states required to express (approximate) a bosonic quantum state. Gaussian states form a class of states which includes coherent states and squeezed states, and they are also easy to simulate classically when evolving under Gaussian bosonic circuits \cite{Bartlett_2002}. Hence, a natural question is whether our proof techniques can be modified to obtain lower bounds on the approximate Gaussian rank of bosonic quantum states.
For instance, one could expect that the low-rank approximation technique used for deriving \cref{thm:lower_bound_csr_informal} may be adapted to the Gaussian rank. The high-level intuition comes from signal analysis and is as follows. In the context of signal analysis, given a signal with fixed frequencies $f(t)=\sum_{k=1}^rc_ke^{\lambda_k t}$, Prony's method \cite{prony1795essai} allows one to recover the frequencies $\lambda_k$ and the coefficients $c_k$. On the other hand, the Fock basis amplitudes of a superposition of $r$ coherent states take the form $f(n)=\sum_{k=1}^rc_k(\alpha_k)^n$, up to a factor $\sqrt{n!}$. Hence, for a single-mode quantum state $\ket\psi$, looking for a coherent state decomposition amounts to retrieving the corresponding frequencies and coefficients by treating its (rescaled) Fock basis amplitudes as a signal, which then leads to the proof technique of \cref{thm:lower_bound_csr_informal}. The analogous problem for Gaussian rank would be to extract the features of a chirp \cite{kundu2021chirp}, i.e.\ a signal with varying frequencies.

Another natural avenue for generalizing \cref{thm:lower_bound_csr_informal} is to extend its proof technique to multimode quantum states, replacing Hankel matrices by Hankel tensors \cite{qi2013hankel}. However, such a generalization may turn out to be challenging, since there is no general equivalent of the Young--Eckart--Mirsky theorem for tensors \cite{kolda2003counterexample}.

In the case of the $n$-mode Fock state $\ket1^{\otimes n}$, we were able to show that $\kappa(\ket1^{\otimes n})=n^{\Omega(\log n)}$ by exploiting a connection to algebraic complexity (see \cref{thm:tensor_fock_lowerbound_csr}), and more precisely to the border complexity of multi-linear formulas for the permanent (see \cref{thm:bordercompperinformal}), generalizing a seminal result of Raz \cite{raz2009multi}. We note that a different generalization was obtained in \cite[Corollary 10]{aaronson2004multilinear}. Since the permanent governs the probabilities of Boson Sampling \cite{aaronson_2010}, \cref{thm:tensor_fock_lowerbound_csr} rules out efficient classical algorithms for simulating Boson Sampling based on coherent state decompositions \cite{Marshall_2023}. For multimode states, only trivial upper bounds on the approximate coherent state rank are known, obtained by taking tensor products of single-mode approximate coherent state decompositions, such as $\kappa(\ket1^{\otimes n})\le2^n$. We conjecture that $\kappa(\ket1^{\otimes n})=2^n$.

Interestingly, both Raz's lower bound and \cref{thm:bordercompperinformal} are also valid for the determinant, i.e.\ (approximate) multi-linear formulas for the determinant must have super-polynomial size. While the permanent relates to bosonic statistics, the determinant appears naturally in fermionic statistics \cite{oszmaniec2022fermion}. We leave as an open question the consequences of our results for fermionic systems.

On a different note, super-polynomial lower bounds on the stabilizer rank in the qubit setting have remained elusive \cite{Bravyi_2016,Peleg_2022,Lovitz2022newtechniques,mehraban2024}. In particular, \cite{mehraban2024} shows that the stabilizer rank of $\ket{T}^{\otimes n}$ is super-polynomial assuming that no polynomial-size circuits for the permanent exists. Here we leverage instead a super-polynomial lower bound on multilinear formula size for the permanent. Given the similarities between the two problems, it is natural to wonder whether the proof techniques of \cref{thm:tensor_fock_lowerbound_csr} for bounding the approximate coherent state rank can be adapted to bound the approximate stabilizer rank. For instance, can we directly compile Boson Sampling using qubits to leverage the algebraic complexity of the permanent to lower bound the stabilizer rank? We also leave this question as an interesting open problem.

\subsection{Acknowledgements}

The authors thank T.\ Vidick for interesting discussions.
U.C.\ thanks J.\ Davis, J.\ Slote, A.\ Deshpande, K.\ A.\ Okoudjou and S.\ Mehraban for inspiring discussions. U.C.\ acknowledges funding from the European Union’s Horizon Europe Framework Programme (EIC Pathfinder Challenge project Veriqub) under Grant Agreement No.\ 101114899.

\section{Preliminaries}
\label{sec:prelim}

In this section, we introduce the necessary preliminary material for our results and proofs. We refer the reader to \cref{appendix:theory} for further background, as well as \cite{Nielsen_Chuang_2010,serafini2023quantum}.

\subsection{Norms, distances, and similarity measures \label{sec:norms-distance}}

We denote the \emph{2-norm} of a vector $\vec{x}=(x_0, x_1, \dots, x_n),\;n\in\mathbb{N}$, as
\begin{equation}
    \norm{\vec{x}}_2 = \left(\sum_{k=0}^n|x_k|^2\right)^{\frac12}.
    \label{equ:2-norm}
\end{equation} 

Let $A$ be an $m$ by $n$ matrix, its Frobenius norm is given by \cite{horn2013matrix}: 
\begin{equation}
    \norm{A}_F = \sqrt{\sum_i^m\sum_j^n|a_{ij}|^2} = \sqrt{\Tr{A^*A}}=\sqrt{\sum_i^{\min\{m,n\}}\sigma_i^2(A)},
\end{equation} with $\sigma_i(A)$ being the singular values of $A$.

We will also employ measures quantifying the similarity between two quantum states. The \emph{trace distance} between two quantum states $\rho$ and $\sigma$ is given by \cite{Nielsen_Chuang_2010}
\begin{equation}
    D(\rho, \sigma) := \frac12 \norm{\rho-\sigma}_1 = \frac{1}{2}\Tr{\sqrt{(\rho-\sigma)^\dagger(\rho-\sigma})},
    \label{equ:trace-distance}
\end{equation} where $\norm{A}_1 = \Tr{\sqrt{A^\dagger A}}$.
The \emph{fidelity} is defined as \cite{Nielsen_Chuang_2010}
\begin{equation}
    F(\rho, \sigma) = \left(\Tr{\sqrt{\sqrt{\rho}\sigma\sqrt{\rho}}}\right)^2.
\end{equation} 
In the specific case when $\rho$ and $\sigma$ are pure states, that is $\rho = \ket{\psi_\rho}\bra{\psi_\rho}$ and $\sigma = \ket{\psi_\sigma}\bra{\psi_\sigma}$, the fidelity simplifies to 
\begin{equation}
    F(\rho, \sigma) = \left|\braket{\psi_\rho|\psi_\sigma}\right|^2.
    \label{equ:fidelity}
\end{equation} 
The \emph{fidelity} and the \emph{trace distance} are closely related. In the case of mixed quantum states, they are qualitatively equivalent measures of closeness due to the following inequalities \cite{Nielsen_Chuang_2010}:
\begin{equation}
    1-\sqrt{F(\rho, \sigma)}\leq D(\rho, \sigma) \leq \sqrt{1-F(\rho, \sigma)}.
\end{equation} 
In the case of pure states $\ket{\psi_\rho}$ and $\ket{\psi_\sigma}$, the two measures become completely equivalent, following the relation \cite{Nielsen_Chuang_2010}
\begin{equation}
    D({\psi_\rho},{\psi_\sigma}) = \sqrt{1-F({\psi_\rho},{\psi_\sigma})}.
    \label{equ:trace-fidelity-equiv}
\end{equation}

\subsection{Approximate coherent state rank}

In this section we recall the definitions of approximate coherent state rank and $\varepsilon$-approximate coherent state rank.

\smallskip

The degree of non-classicality \cite{gherke_2012, Sperling_2015} measures the minimal number of coherent states necessary to express as a superposition a given pure bosonic quantum state. In the context of quantum information theory, it is convenient to introduce the following approximate version:

\begin{definition}[Approximate coherent state rank \cite{Marshall_2023}]
The approximate coherent state rank of an $m$-mode state $\ket{\psi}$, denoted $\kappa({\psi})$, is the smallest $k$ such that for all $\varepsilon>0$, there exists a state $\ket{\phi_\varepsilon}=\sum_{j=1}^kc_j\ket{\bm\alpha_j}$, where $c_j\in\mathbb C$ and $\ket{\bm\alpha_j}$ are $m$-mode coherent states, such that $F(\psi,\phi_\varepsilon)\geq1-\varepsilon$.
\label{def:approx-coherent-rank}
\end{definition}
 
We also use a fine-grained version of the approximate coherent state rank:
 
\begin{definition}[$\varepsilon$-approximate coherent state rank \cite{Marshall_2023,upreti2025bounding}]
Let $\varepsilon>0$. The $\varepsilon$-approximate coherent state rank of a state $\ket{\psi}$, denoted $\kappa_\varepsilon({\psi})$, is the smallest $k$ such that there exists a state $\ket{\phi}=\sum_{j=1}^kc_j\ket{\alpha_j}$, where $c_j\in\mathbb C$ and $\ket{\alpha_j}$ are coherent states, such that $F(\psi,\phi)\geq1-\varepsilon$.
\label{def:eps-approx-coherent-rank}
\end{definition}

The approximate coherent state rank is directly obtained from the $\varepsilon$-approximate coherent state rank via $\kappa=\sup_{\varepsilon>0}\kappa_\varepsilon=\lim_{\varepsilon\to0^+}\kappa_\varepsilon$.

\subsection{Passive linear unitaries and vacuum projection}
\label{sec:linearoptics}

In this section we introduce concepts stemming from quantum optics.

\smallskip

A passive linear unitary (or linear-optical network or unitary interferometer in the context of quantum optics) on $m$ modes is described by an $m \times m$ unitary matrix $U$. The corresponding unitary operator $\hat{U}$ in Fock space transforms the creation operators (see \cref{appendix:theory}) according to the relation:
\begin{equation}
    \hat{U} \hat{a}_k^\dagger \hat{U}^\dagger = \sum_{l=1}^m U_{lk} \hat{a}_l^\dagger.
\end{equation}
Fundamental properties of these transformations are that they conserve the total number of bosons, they map the vacuum state $\ket{\bm0}$ to itself, and more generally they map a tensor product of coherent states $\ket{\bm{\alpha}} = \ket{\alpha_1} \otimes \dots \otimes \ket{\alpha_m}$ to another tensor product of coherent states $\ket{\bm{\alpha}'} = \ket{U \bm{\alpha}}$ \cite{ferraro2005}. 

Now, we consider the projection onto the vacuum state for all modes except the first, represented by $(\mathbb{I} \otimes \bra{0}^{\otimes(m-1)})$. Applying this to a multimode coherent state $\ket{\bm{\alpha}}$ yields a single-mode coherent state $\ket{\alpha_1}$ scaled by a sub-normalization factor:
\begin{equation}
    (\mathbb{I} \otimes \bra{0}^{\otimes (m-1)}) \ket{\alpha_1, \dots, \alpha_m} = \exp\left(-\frac{1}{2} \sum_{j=2}^m |\alpha_j|^2\right) \ket{\alpha_1}.
\end{equation}
By linearity, if a state $\ket{\bm{\phi}}$ is a superposition of $r$ multimode coherent states, its reduction to the first mode remains a superposition of at most $r$ single-mode coherent states.

Finally, as alluded to in the introduction, Fock basis amplitudes of passive linear unitaries relate to the permanent \cite{scheel2004,aaronson_2010}. In particular:
\begin{equation}\label{eq:linktoper}
    \bra1^{\otimes m}\hat U\ket1^{\otimes m}=\mathrm{Per}(U).
\end{equation}

\subsection{Low-rank approximation theory}

In this section, we formally introduce tools from low-rank approximation theory.
 
The Vandermonde and Hankel matrices are two classes of structured matrices arising in approximation theory. A Vandermonde matrix $V$ is generated by geometric progressions of some vector $\vec{x}=(x_0, x_1, \dots, x_{2N})$: 
\begin{equation}
    V(\vec{x}) =(x_i^j)_{0\le i,j\le N}.
    \label{equ:vandermonde}
\end{equation}
A Hankel matrix of the same vector is given by
\begin{equation}
    H(\vec{x}) = (x_{i+j})_{0\le i,j\le N},
    \label{equ:hankel}
\end{equation}
and is in particular constant along its anti-diagonals. In the specific case where the vector $\vec{x}=(x_0, x_1, \dots, x_{2N})$ is a sum of $r$ exponentials, i.e.,
\begin{equation}
    x_k = \sum_{i=1}^rd_i\lambda_i^k,
\end{equation}
these two forms are fundamentally connected through a factorization of the form
\begin{equation}\label{eq:HVDV}
    H(\vec{x}) = V(\bm\lambda)^T D(\bm d) V(\bm\lambda),
\end{equation} where $V(\bm\lambda)$ is a $r\times(N+1)$ Vandermonde matrix and $D=\mathrm{Diag}(d_1,\dots,d_r)$. Moreover, a Hankel matrix $H(\vec{x})$ has rank $r$ if and only if such a factorization exists, with all $\lambda_i$'s distinct and $d_i\neq0$ \cite{boley11998vandermonde}. This structure underlies Prony's method \cite{prony1795essai} and low-rank approximation theory \cite{knirsch2021,Shen2022,gillard2022hankel}.

Low-rank approximation refers to the process of approximating a given matrix by another matrix of lower rank while preserving as much information as possible. We can express the problem as minimizing over $\tilde{M}$ the Frobenius norm $\norm{M-\tilde{M}}_F$ with $\text{rank}(\tilde{M})\leq r$. Low-rank approximations are typically used in signal processing or data compression tasks where the rank is related to the model fitting the data. Some techniques from signal analysis will be adapted to our specific problem.

A fundamental theorem in low-rank approximation due to C.\ Eckart, G.\ Young \cite{Eckart_Young_1936}, and L.\ Mirsky \cite{mirsky1960symmetric} provides the optimal method for reducing the rank of a matrix while minimizing the error in the Frobenius norm.
 
\begin{theorem}[Young--Eckart--Mirsky]
\label{thm:YEM}
Let $A\in\mathbb{C}^{m\times n}$ be a complex matrix with singular value decomposition
\[
A = U\Sigma V^*,
\]
where $U\in\mathbb{C}^{m\times m}$ and $V\in\mathbb{C}^{n\times n}$ are unitary, and
\[
\Sigma = \operatorname{diag}(\sigma_1,\sigma_2,\dots,\sigma_R,0,\dots,0),
\]
padded with zeros, with singular values $\sigma_1\ge\sigma_2\ge\cdots\ge\sigma_R>0$ and $R=\operatorname{rank}(A)$. For each $r$ with $0<r\le R$, let
\[
\Sigma_{(r)} = \operatorname{diag}(\sigma_1,\dots,\sigma_r,0,\dots,0),
\]
and define the truncated SVD approximation
\[
A_r := U\,\Sigma_{(r)}\,V^*.
\]

Then, for all $r<R$, the following holds:

\begin{equation}
\|A - A_r\|_F = \min_{\operatorname{rank}(B)\,\le r} \|A-B\|_F = \sqrt{\sigma_{r+1}^2+\cdots+\sigma_R^2} = \sqrt{\sum_{l>r}^R\sigma_l^2},
\end{equation}
where $\|\cdot\|_F$ denotes the Frobenius norm, and where the right hand side is zero for $r\ge R$. The matrix $A_r$ achieving this minimum is unique if and only if $\sigma_r>\sigma_{r+1}$. For all $r\ge R$, the minimum error is zero, achieved by $A$ itself.
\end{theorem}

\section{Single-mode lower bounds on the approximate coherent state rank}

In this section, we obtain lower bounds for the approximate coherent state rank of single-mode bosonic quantum states.

\subsection{Generic lower bounds}

In this section we prove the formal version of \cref{thm:lower_bound_csr_informal}:

\begin{theorem}[Generic lower bounds on the $\varepsilon$-approximate coherent state rank]
Let $\ket{\psi} = \sum_{n=0}^\infty \psi_n\ket{n}$ be a single-mode quantum state, let $N\in\mathbb N$ be a truncation parameter, and let $H_N(\psi)$ be the Hankel matrix built from the first $2N+1$ rescaled coefficients $\{\sqrt{n!}\psi_n\}_{n=0}^{2N}$ of $\ket{\psi}$, i.e.\ $(H_N(\psi))_{ij}=\sqrt{(i+j)!}\,\psi_{i+j}$, for $0\leq i ,j \leq N$. Then, 
\begin{equation}
\label{eq:genlow}
    \varepsilon < \frac{1}{2(N+1)(2N)!}\sum_{l=r+1}^{N+1}\sigma_l(H_N(\psi))^2\quad\Rightarrow\quad\kappa_\varepsilon(\psi)>r,
\end{equation}
for all $r\le N\in\mathbb N$.
\label{thm:lower_bound_appox_coherent_stellar}
\end{theorem}

\begin{proof}
\label{proof:lower_bound_appox_coherent_stellar}

Let
\begin{equation}
\ket{\phi} = \sum_{k=1}^r c_k \ket{\alpha_k}
\end{equation}
be an arbitrary superposition of $r$ distinct coherent states. We have
\begin{equation}\label{eq:relnorm}
    \|\ket\psi-\ket\phi\|_2\le\sqrt{2(1-F(\psi,\phi))}.
\end{equation}
Taking the contrapositive in \cref{eq:genlow}, we aim to show that the fidelity between the state $\ket\psi$ and any superposition of $r$ coherent states is upper bounded by
$1-\frac{1}{2(N+1)(2N)!}\sum_{l>r}\sigma_l(H_N(\psi))^2$, where $H_N(\psi)=(\sqrt{(i+j)!}\,\psi_{i+j})_{0\le i,j\le N}$ is the Hankel matrix associated to the rescaled Fock amplitudes of the state $\ket\psi$. By \cref{eq:relnorm}, it is enough to show
\begin{equation}\label{eq:simplerbound}
    \|\ket\psi-\ket\phi\|_2^2\ge\frac{1}{(N+1)(2N)!}\sum_{l=r+1}^{N+1}\sigma_l(H_N(\psi))^2.
\end{equation}
To prove this bound, we examine the representation of $\ket\phi$ in the Fock basis: we define
\begin{equation}
    \phi_n \coloneqq \braket{n|\phi} = \sum_{k=1}^r c_k e^{-|\alpha_k|^2/2} \frac{\alpha_k^n}{\sqrt{n!}}.
\end{equation}
Rescaling these coefficients by $\sqrt{n!}$ yields 
\begin{equation}
    \sqrt{n!}\,\phi_n=\sum_{k=1}^rd_k\alpha_k^n,
\end{equation}
where we have set $d_k\coloneqq c_k e^{-|\alpha_k|^2/2}$. By \cref{eq:HVDV}, the associated Hankel matrix $H_N(\phi)=(\sqrt{(i+j)!}\,\phi_{i+j})_{0\le i,j\le N}$ takes the form
\begin{equation}
    H_N(\phi)=V(\bm\alpha)^TD(\bm d)V(\bm\alpha),
\end{equation}
for all $N\ge r$, where $V(\bm\alpha)$ is an $r\times(N+1)$ Vandermonde matrix and $D(\bm d)=\mathrm{Diag}(d_1,\dots,d_r)$. In particular, the rank of $H_N(\phi)$ is at most $r$, so by the Young--Eckart--Mirsky theorem (\cref{thm:YEM}),
\begin{equation}\label{eq:usingYEM}
    \|H_N(\psi)-H_N(\phi)\|_F^2\ge\sum_{l=r+1}^{N+1}\sigma_l(H_N(\psi))^2,
\end{equation}
where $\|\cdot\|_F$ denotes the Frobenius norm.

Finally, we relate the Euclidean distance of states to the Frobenius distance of the associated Hankel matrices via the following technical result, whose proof is deferred to \cref{appendix:proofs}:

\begin{lemma}
    Let $N\ge0$ and let $\ket\psi= \sum_{n=0}^\infty \psi_n\ket{n}$ and $\ket\phi= \sum_{n=0}^\infty \phi_n\ket{n}$ be two single-mode pure states, with associated Hankel matrices $H_N(\psi)=(\sqrt{(i+j)!}\,\psi_{i+j})_{0\le i,j\le N}$ and $H_N(\phi)=(\sqrt{(i+j)!}\,\phi_{i+j})_{0\le i,j\le N}$, respectively. Then, 
    \begin{equation}
        \|\ket\psi-\ket\phi\|_2^2\ge\frac{1}{(N+1)(2N)!}\|H_N(\psi)-H_N(\phi)\|_F^2.
    \end{equation}
    \label{lemma:norm_H}
\end{lemma} 

Combining \cref{lemma:norm_H} with \cref{eq:usingYEM} yields \cref{eq:simplerbound}, and concludes the proof.

\end{proof}

We note that the Young--Eckart--Mirsky theorem used in \cref{eq:usingYEM} in the proof of \cref{thm:lower_bound_appox_coherent_stellar} also shows the existence of a low-rank matrix saturating the bound. However, the optimal matrix is not necessarily a Hankel matrix, indicating that the bound on $\varepsilon$ in \cref{thm:lower_bound_appox_coherent_stellar} is not tight in general.
Moreover, the factorial term $(2N)!$ in the bound grows very rapidly, which results in a loose bound for high Fock numbers. To sharpen the result, we introduce in \cref{app:optib} an additional free rescaling parameter $b$, which aims to counterbalance the factorial scaling by re-weighting the importance of higher Fock states, leading to a tighter family of lower bounds (see \cref{thm:optimised_b}) which may be optimized numerically.

\subsection{Characterization of single-mode states of finite approximate coherent state rank}
\label{sec:statespec}

A remarkable consequence of \cref{thm:lower_bound_appox_coherent_stellar} is the following result:

\begin{theorem}[Characterization of single-mode states of finite approximate coherent state rank]
Let $\ket\psi=\sum_{n=0}^{+\infty}\psi_n\ket n$ be a single-mode quantum state and let $r>0$. Then, $\kappa(\psi)=r$ if and only if the sequence $(\sqrt{n!}\,\psi_n)_{n\in\mathbb N}$ follows a linear recurrence relation of order $r$. Equivalently, $\ket\psi$ can be written as a superposition of $k$ displaced core states
\begin{equation}
    \ket\psi=\sum_{j=1}^kc_j\hat D(\alpha_j)\ket{C_j},
\end{equation}
where all displacement amplitudes $\alpha_j$ are distinct and where each core state $\ket{C_j}$ has highest Fock number equal to $n_j$ such that $\sum_{j=1}^k(n_j+1)=r$.
\label{thm:LLR}
\end{theorem}

\begin{proof}
We first show the forward implication. Let $H_N(\psi)=(\sqrt{(i+j)!}\,\psi_{i+j})_{0\le i,j\le N}$, for all $N\in\mathbb N$. The condition $\kappa(\psi)=r$ implies $\kappa_\varepsilon(\psi)\le r$ for all $\varepsilon>0$, so taking the contrapositive in \cref{eq:genlow} we obtain
\begin{equation}
    \frac{1}{2(N+1)(2N)!}\sum_{l=r+1}^{N+1}\sigma_l(H_N(\psi))^2\le\varepsilon,
\end{equation}
 for all $\varepsilon>0$. Taking the limit when $\varepsilon\to0$ yields $\sigma_l(H_N(\psi))=0$ for all $l>r$, so the rank of $H_N(\psi)$ is at most $r$, for all $N\in\mathbb N$.
 This implies that the sequence $(\sqrt{n!}\,\psi_n)_{n\in\mathbb N}$ follows a linear recurrence relation of order at most $r$ \cite[Section I.3, Lemma III]{salem1963algebraic}. Hence, there exist $k\le r$, polynomials $P_1,\dots,P_k$ of degree $n_1,\dots,n_k$ with $\sum_{j=1}^k(n_j+1)\le r$ and distinct complex numbers $\alpha_1,\dots,\alpha_k$ such that \cite{rosen1999discrete}:
\begin{equation}
    \sqrt{n!}\,\psi_n=\sum_{j=1}^kP_j(n)\alpha_j^n.
\end{equation}
Equivalently using $\ket\alpha=e^{-|\alpha|^2/2}\sum_{n=0}^{+\infty}\frac{\alpha^n}{\sqrt{n!}}\ket n$, we have $\ket\psi=\sum_{j=1}^ke^{|\alpha_j|^2/2}P_j(\hat n)\ket{\alpha_j}$, where $\hat n=\hat a^\dag\hat a$ is the number operator. Writing each polynomial $P_j(\hat n)=Q_j(\hat a^\dag,\hat a)$ in normal order (with all creation operators to the left of annihilation operators) and using $\hat a\ket\alpha=\alpha\ket\alpha$, we obtain
 \begin{align}
    \ket\psi&=\sum_{j=1}^ke^{|\alpha_j|^2/2}Q_j(\hat a^\dag,\hat a)\ket{\alpha_j}\\
    &=\sum_{j=1}^ke^{|\alpha_j|^2/2}Q_j(\hat a^\dag,\alpha_j)\ket{\alpha_j}\\
    &=\sum_{j=1}^ke^{|\alpha_j|^2/2}Q_j(\hat a^\dag,\alpha_j)\hat D(\alpha_j)\ket{0}\\
    &=\sum_{j=1}^k\hat D(\alpha_j)\left[e^{|\alpha_j|^2/2}Q_j(\hat a^\dag+\alpha_j^*,\alpha_j)\right]\ket{0},
\end{align}
 where we used $\hat a^\dag\hat D(\alpha)=\hat D(\alpha)(\hat a^\dag+\alpha^*)$ with $\hat D$ the displacement operator. Setting for all $j\in\{1,\dots,k\}$
 \begin{equation}
 c_j\coloneqq\|e^{|\alpha_j|^2/2}Q_j(\hat a^\dag+\alpha_j^*,\alpha_j)\ket{0}\|\qquad\text{and}\qquad\ket{C_j}\coloneqq\frac1{c_j}e^{|\alpha_j|^2/2}Q_j(\hat a^\dag+\alpha_j^*,\alpha_j)\ket{0},
 \end{equation}
 we obtain $\ket\psi=\sum_{j=1}^kc_j\hat D(\alpha_j)\ket{C_j}$. Since $Q_j(\hat a^\dag+\alpha_j^*,\alpha_j)$ is a polynomial in $\hat a^\dag$ of degree at most $n_j$, the state $\ket{C_j}$ is a core state with highest Fock number at most $d_j\le n_j$, for all $j\in\{1,\dots,k\}$.

Now, a core state with highest Fock number $n$ has an approximate coherent state rank upper bounded by $n+1$ \cite[Theorem 1]{Marshall_2023}. Hence, a superposition of $k$ displaced core states $\hat D(\alpha_j)\ket{C_j}$ with highest Fock number $d_j$ has approximate coherent state rank at most $\sum_{j=1}^k(d_j+1)$. Since the approximate coherent state rank of $\ket\psi=\sum_{j=1}^kc_j\hat D(\alpha_j)\ket{C_j}$ is equal to $r$, this implies $r\le\sum_{j=1}^k(d_j+1)\le\sum_{j=1}^k(n_j+1)\le r$, so $d_j=n_j$ for all $j\in\{1,\dots,k\}$ and $\sum_{j=1}^k(n_j+1)=r$.

\smallskip

We now show the reverse implication. Reciprocally, if $\ket\psi=\sum_{j=1}^kc_j\hat D(\alpha_j)\ket{C_j}$ is a finite superposition of displaced core states with distinct $\alpha_j$ and $\sum_{j=1}^k(n_j+1)=r$, the same reasoning based on \cite[Theorem 1]{Marshall_2023} shows that $\kappa(\psi)\le r$. Let us assume for the sake of contradiction that $\kappa(\psi)<r$. Then, by the forward implication above, there exists a decomposition $\ket\psi=\sum_{j=1}^{k'}c_j'\hat D(\alpha_j')\ket{C_j'}=\sum_{j=1}^kc_j\hat D(\alpha_j)\ket{C_j}$ with $k'<r$ and $\sum_{j=1}^{k'}(n_j'+1)<r$. We consider the function $F_\psi^\star:z\mapsto e^{\frac12|z|^2}\langle z^*|\psi\rangle$ (the stellar function of $\ket\psi$, see \cref{appendix:theory}). We obtain
\begin{equation}
F_\psi^\star(z)=\sum_{j=1}^k\tilde P_j(z)e^{\alpha_jz}=\sum_{j=1}^{k'}\tilde P_j'(z)e^{\alpha_j'z},
\end{equation}
for all $z\in\mathbb C$, where $\tilde P_j$ and $\tilde P_j'$ are non-zero polynomials of degree $n_j$ and $n_j'$, respectively, which can be directly computed from the two decompositions of $\ket\psi$. Taking the difference and grouping the exponential terms, we obtain an identity of the form $\sum_mR_m(z)e^{\gamma_mz}=0$, where $\gamma_m$ ranges over the union of all distinct exponents $\{\alpha_j\}\cup\{\alpha_j'\}$ and where $R_m$ are polynomials (either of the form $\tilde P_i$, $\tilde P_j'$ or $\tilde P_i-\tilde P_j'$, depending on whether the corresponding exponent is common to both decompositions). The growth when $z\to\infty$ along rays in the complex plane given by $-\arg(\gamma_m)$ then implies that all the polynomials $R_m$ must vanish. Hence, $k=k'$ and (up to a permutation) $\alpha_j=\alpha_j'$ for all $j\in\{1,\dots,k\}$. We thus obtain $\sum_{j=1}^k(n_j'+1)<\sum_{j=1}^k(n_j+1)$. This implies in turn that there exists an index $j$ such that $n_j'<n_j$ so $\tilde P_j-\tilde P_j'$ is not identically zero, which contradicts the fact that $R_m$ vanishes for all $m$. This shows that $\kappa(\psi)=r$.
\end{proof}

\cref{thm:LLR} provides a complete characterization of single-mode states of finite approximate coherent state rank and indicates that most states have infinite approximate coherent state rank. 

\smallskip

More generally, lower bounds on the $\varepsilon$-approximate coherent state rank of a state $\ket\psi$ can be obtained by studying its associated Hankel matrices $H_N(\psi)$, for $N\in\mathbb N$. We illustrate this technique for specific quantum states, starting with core states:

\begin{corollary}[Approximate coherent state rank of core states]\label{coro:coreFockACSR}
   Let $\psi$ be a finite superposition of Fock states with highest Fock number $n$. Then, its approximate coherent state rank is given by $\kappa(\psi)=n+1$. Moreover, for $\ket{\psi}=\ket n$, we have $\kappa_\varepsilon(\ket n)=n+1$ for all $\varepsilon < \frac{n!}{2(n+1)(2n)!}$.
\end{corollary}

\begin{proof} The first part of the result is a direct consequence of \cref{thm:LLR} for $k=1$, $\alpha_j=0$ and $\ket{C_j}=\ket\psi$, with $n_j=n$.

Consider the Fock state $\ket\psi=\ket n$. Its associated Hankel matrix of size $n+1$ is given by
\begin{equation}
        H_n(\psi) =
        \begin{pmatrix}
        0 & 0 & \cdots & 0 & \sqrt{n!} \\
0 & 0 & \cdots & \sqrt{n!} & 0 \\
\vdots & \vdots & \ddots & \vdots & \vdots \\
0 & \sqrt{n!} & \cdots & 0 & 0 \\
\sqrt{n!} & 0 & \cdots & 0 & 0
\end{pmatrix}.
        \label{equ:fock-hankel}
\end{equation} This is an anti-diagonal matrix with all singular values equal to $\sqrt{n!}$.
From \cref{thm:lower_bound_appox_coherent_stellar}, we know that $\kappa_{\varepsilon}(\psi)>r$ for all $\varepsilon>0$ such that
\begin{equation}
    \varepsilon < \frac{1}{2(n+1)(2n)!}\sum_{l=r+1}^{n+1}\sigma_l(H_n(\psi))^2,
    \label{equ:fock_eps_threshold}
\end{equation}
so taking $r=n$ we obtain $\kappa_\varepsilon(\ket n)=n+1$ whenever
\begin{equation}
    \varepsilon < \frac{n!}{2(n+1)(2n)!}.
\end{equation}
Applying the optimized \cref{thm:optimised_b} instead of \cref{thm:lower_bound_appox_coherent_stellar} leads to a much stronger bound, as we detail in \cref{app:optib}.
\end{proof}
 
Note that by \cref{thm:LLR1}, the Fock basis amplitudes of $\ket\psi$ must follow a linear recurrence relation of order at most $n+1$, which in this pathological case reads $\psi_{p+n+1}=0$, for all $p\in\mathbb N$.

\smallskip

We now consider squeezed states:

\begin{corollary}[Approximate coherent state rank of single-mode squeezed states]
    The $\varepsilon$-approximate coherent state rank of a squeezed state $\kappa_\varepsilon(\ket{\xi})$ with $\xi\neq0$ goes to infinity as $\varepsilon$ goes to $0$. Equivalently, \mbox{$\kappa(\ket{\xi})=+\infty$} for all $\xi\neq0$.
    \label{coro:no_finite_decomp_squeezed2}
\end{corollary}

\begin{proof}
    \label{proof:no_finite_decomp_squeezed}
    The Fock expansion of a squeezed state is $\ket{\xi} = \frac{1}{\sqrt{\cosh(r)}}\sum_{n=0}^\infty(-e^{i\phi}\tanh(r))^n\frac{\sqrt{(2n)!}}{2^nn!}\ket{2n}$, where $\xi=re^{i\phi}\in\mathbb C$. Setting $\lambda = -e^{i\phi}\tanh(r)$:
    \begin{equation}
        \ket{\xi} = \frac{1}{\mathcal{N}}\sum_{n=0}^\infty\frac{\lambda^n\sqrt{(2n)!}}{2^nn!}\ket{2n} =  \frac{1}{\mathcal{N}}\sum_{n=0}^\infty \xi_{2n}\ket{2n},
    \end{equation} where $\mathcal{N}$ accounts for normalization, and where $ \xi_{2n} = \frac{\lambda^n\sqrt{(2n)!}}{2^nn!}$ and $\xi_{2n+1}=0$ are the coefficients. We show that this sequence does not satisfy any linear recurrence relation, unless $\xi=0$, which by \cref{thm:LLR} implies that $\kappa(\ket\xi)> r$ for all $r\in\mathbb N$. It is enough to show the following:

\begin{lemma}
    For $\lambda \neq 0$, the sequence $s_n = \frac{\lambda^n\sqrt{(2n)!}}{2^n\sqrt{n!}}$ does not satisfy a linear recurrence relation with constant coefficients of any order.
    \label{lem:sn_no_LRR}
\end{lemma}

From \cite[Proposition 4.2.5]{Stanley2012EnumerativeCombinatorics1}, if $F(x)$ and $G(x)$ are rational power series, then so is their Hadamard product $F*G$. Using \cite[Theorem 4.1.1]{Stanley2012EnumerativeCombinatorics1}, this means that if two sequences $f(n)$ and $g(n)$ each satisfy a linear recurrence relation with constant coefficients, then so does their product $h(n) = f(n)g(n)$.

Let us therefore consider $h_n = s_n^2 = \frac{\lambda^{2n}(2n)!}{2^{2n}n!}$. We compute the ratio
    \begin{equation}
    R(n)=\frac{h_{n+1}}{h_n} = \frac{\frac{\lambda^{2n+2}(2n+2)!}{2^{2n+2}(n+1)!}}{\frac{\lambda^{2n}(2n)!}{2^{2n}n!}} = \frac{\lambda^2}{4}\frac{(2n+2)(2n+1)}{n+1} = \frac{\lambda^2}{2}(2n+1).
    \label{equ:hn_rec_rel}
    \end{equation}
For the sake of contradiction, let us assume that $h_n$ satisfies a linear recurrence relation with constant coefficients, then for some $k\in \mathbb{N}^*$, using the relation in \cref{equ:hn_rec_rel},
\begin{align}
    h_{n+k} &= c_{k-1} h_{n+k-1} + \dots + c_1 h_{n+1} + c_0 h_{n} \\
    \Leftrightarrow \left( \prod_{j=0}^{k-1} R(n+j) \right) h_{n}& = \left[ c_{k-1}\prod_{j=0}^{k-2} R(n+j)  + \dots + c_1R(n) + c_0 \right]h_{n},
\end{align} 
so that 
\begin{equation}
    \left[ \prod_{j=0}^{k-1} R(n+j)  -  c_{k-1}\prod_{j=0}^{k-2} R(n+j)  - \dots - c_1R(n) - c_0 \right]h_{n} = 0.
    \label{equ:deg-k-poly-0}
\end{equation}
By construction, the left-hand side of \cref{equ:deg-k-poly-0} is a polynomial of degree $k$ in $n$. It therefore has at most $k$ roots, but from \cref{equ:deg-k-poly-0} it must be zero for all $n\in \mathbb{N}$, which implies that $\lambda=0$.

Therefore, for $\xi\neq0\Leftrightarrow\lambda\neq0$, $h_n$ does not satisfy any linear recurrence relation with constant coefficients and thus neither does $s_n$, and so neither does $\xi_n$.
\end{proof}

\section{Multimode lower bounds on the approximate coherent state rank}

In this section, we derive lower bounds for the approximate coherent state rank of multimode bosonic quantum states.

\subsection{From single-mode to multimode lower bounds}

We first provide a lower bound on the approximate coherent state rank of a superposition of multimode core states:

\begin{theorem}[Lower bound on the approximate coherent state rank of multimode core states]
   Let $\ket{\psi}$ be a finite superposition of multimode Fock states,
   \begin{equation}
    \ket{\bm\psi} = \sum_{|\vec{n}| \leq n} c_n \ket{\vec{n}}, \quad \text{with } |\vec{n}| =  \sum_{i=1}^{m} n_i.
    \end{equation} If $\ket{\bm\psi}$ has a non-vanishing component in the $n$-boson subspace (i.e.\ there exists at least one $\vec{n}$ with $|\vec{n}| = n$ such that $c_{\vec{n}} \neq 0$), then
    \begin{equation}
        \kappa(\bm\psi)\ge n+1.
    \end{equation}
        \label{thm:multimode_core_state_lowerbound_csr2}
\end{theorem}

\begin{proof}
According to \cref{lem:unitary_existence} below, there exists an $m \times m$ passive linear unitary $U$ such that the amplitude $d_n$ of the fully bunched state $\ket{n, 0, \dots, 0}$ at the output is non-zero. Let $\hat{U}$ be the corresponding unitary operator.

We decompose the output state by projecting the last $m-1$ modes onto the vacuum. Let $\ket{C} = (\mathbb{I} \otimes \bra{0}^{\otimes (m-1)}) \hat{U} \ket{\bm\psi}$ represent the state in the first mode. Because passive linear unitaries preserve total boson number, this core state takes the form $\ket{C} = \sum_{k=0}^n d_k \ket{k}$. In particular, due to \cref{lem:unitary_existence}, $d_n \neq 0$. We can then write:
\begin{equation*}
    \hat{U}\ket{\bm\psi} = \ket{C} \otimes \ket{0}^{\otimes (m-1)} + \ket{\perp},
\end{equation*}
where $(\mathbb{I} \otimes \bra{0}^{\otimes (m-1)})\ket{\perp} = 0$.

On the other hand, given an $m$-mode superposition of $r$ coherent states $\ket{\bm\phi} = \sum_{k=1}^r c_k \ket{\bm\alpha_k}$, we can decompose it as $\ket{\bm\phi} = [(\mathbb I \otimes \bra 0^{\otimes (m-1)}) \ket{\bm\phi}] \otimes \ket 0^{\otimes (m-1)} + \ket{\perp'}$. This gives:
\begin{align*}
    \left\|\hat U \ket{\bm\psi} - \ket{\bm\phi} \right\|_2^2 &= \left\| \left( \ket C \otimes \ket 0^{\otimes (m-1)} + \ket \perp \right) - \left( [(\mathbb I \otimes \bra 0^{\otimes (m-1)}) \ket{\bm\phi}] \otimes \ket 0^{\otimes (m-1)} + \ket{\perp'} \right) \right\|_2^2 \\
    &= \left\| \ket C \otimes \ket 0^{\otimes (m-1)} - [(\mathbb I \otimes \bra 0^{\otimes (m-1)}) \ket{\bm\phi}] \otimes \ket 0^{\otimes (m-1)} \right\|_2^2 + \left\| \ket \perp - \ket{\perp'} \right\|_2^2 \\
    &\ge  \left\| \ket C \otimes \ket 0^{\otimes (m-1)} - [(\mathbb I \otimes \bra 0^{\otimes (m-1)}) \ket{\bm\phi}] \otimes \ket 0^{\otimes (m-1)} \right\|_2^2.
\end{align*}
Writing $\ket{\phi_n} \coloneqq (\mathbb I \otimes \bra 0^{\otimes (m-1)}) \ket{\bm\phi}$, we obtain $\ket{\phi_n} = \sum_{k=1}^r c_k (\mathbb I\otimes\bra0^{\otimes(m-1)})\ket{\bm\alpha_k}$, so $\ket{\phi_n}$ is a (sub-normalised) superposition of at most $r$ single-mode coherent states. Using the unitary invariance of the norm:
\begin{align}
   \left\| \ket{\bm\psi} - \hat U^\dagger \ket{\bm\phi} \right\|_2^2 &=  \left\| \hat U \ket{\bm\psi} - \ket{\bm\phi} \right\|_2^2 \nonumber \\
    &\ge \left\|  \ket C \otimes \ket 0^{\otimes (m-1)} - \ket{\phi_n} \otimes \ket 0^{\otimes (m-1)}\right\|_2^2 \\
    &= \|\ket{C}-\ket{\phi_n}\|_2^2.
\end{align}

Let $\ket{\bm\chi}$ be an $m$-mode superposition of $r$ coherent states. Setting $\ket{\bm\phi}=\hat U\ket{\bm\chi}$, $\ket{\bm\phi}$ is also an $m$-mode superposition of $r$ coherent states, since passive linear unitaries map tensor products of coherent states to tensor products of coherent states, and we have
\begin{equation}
    \left\| \ket{\bm\psi} - \ket{\bm\chi} \right\|_2^2\ge\|\ket{C}-\ket{\phi_n}\|_2^2,
\end{equation}
and thus $\kappa(\bm\psi)\ge\kappa(\ket C)$. Now $\ket{C}$ is a single-mode core state with highest Fock number $n$, so \cref{cor:fock_lowerbound_csr} gives $\kappa(\ket C)=n+1$ and thus $\kappa(\bm\psi)\ge n+1$.
\end{proof}

\begin{lemma}[Bunching unitary for multimode core states]
Let $\ket{\bm\psi} = \sum_{|\vec{n}| \leq n} c_n \ket{\vec{n}}, \; \text{with } |\vec{n}| =  \sum_{i=1}^{m} n_i$ 
be an $m$-mode core state with total boson truncation $n$. If the $n$-boson component is non-zero, there exists a passive Gaussian unitary $\hat U$ such that the reduced state in the first mode, $\ket{\psi_\mathrm{red}} = (\mathbb{I} \otimes \bra{0}^{\otimes m-1}) \hat{U} \ket{\bm\psi}$, has a non-zero amplitude for the Fock state $\ket{n}$.
\label{lem:unitary_existence}
\end{lemma}

\begin{proof}
    To begin with, let us decompose the state by total boson number:
    \begin{equation}
        \ket{\bm\psi} = \sum_{k=0}^n \ket{\bm\psi^{(k)}}.
    \end{equation} 
    Since passive linear unitaries conserve the total boson number, the action of any unitary $\hat{U}$ on $\ket{\bm\psi}$ can be viewed independently on each $k$-boson subspace.
    To prove the existence of a $\ket{n}$ component in the first output mode, we only need to consider the $n$-boson part, $\ket{\bm\psi^{(n)}}$. 
    
    Following \cite{aaronson_2010}, we map the $n$-boson component to a homogeneous polynomial.
    Indeed, the $n$-boson component of the state, $\ket{\bm\psi^{(n)}}$, can be expressed in terms of the input creation operators $\hat{a}_j^\dagger$ acting on the vacuum $\ket{\vec{0}}$:
    \begin{equation}
        \ket{\bm\psi^{(n)}} = \sum_{|\vec{n}| = n} c_{\vec{n}} \frac{(\hat{a}_1^\dagger)^{n_1} \dots (\hat{a}_m^\dagger)^{n_m}}{\sqrt{\prod_j n_j!}}  \ket{\vec{0}} = P(\hat{a}_1^\dagger, \dots, \hat{a}_m^\dagger)\ket{\vec{0}},
        \label{equ:multimode_core_2nd_quant}
    \end{equation} 
    where $\vec{n} = (n_1, \dots, n_m)$ and $|\vec{n}|=\sum_k n_k$. We define a homogeneous polynomial $P(z_1, \dots, z_m)$ of degree $n$ by replacing the creation operators $\hat{a}_j^\dagger$ with complex variables $z_j$ in \cref{equ:multimode_core_2nd_quant}:
    \begin{equation}
        P(z_1, \dots, z_m) = \sum_{|\vec{n}| = n} c_{\vec{n}}\frac{z_1^{n_1} \dots z_m^{n_m}}{\sqrt{\prod_j n_j!}}.
    \end{equation}
    Since $\ket{\bm\psi^{(n)}} \neq 0$, there exists at least one $c_{\vec{n}} \neq 0$, implying that $ P(z_1, \dots, z_m)$ is a non-zero homogeneous polynomial.

    A passive linear unitary described by a unitary matrix $U$ evolves the input creation operators into output modes (denoted here by $\hat{b}_i^\dagger$) according to the forward transformation: 
    \begin{equation}
        \hat{U} \hat{a}_j^\dagger \hat{U}^\dagger = \sum_{i=1}^m U_{ij} \hat{b}_i^\dagger.
    \end{equation} 
    Because a passive linear unitary does not affect the total boson number, and $\hat{U}\ket{\vec{0}} = \ket{\vec{0}}$, the evolved $n$-boson state can be directly written by substituting this transformation into \cref{equ:multimode_core_2nd_quant}:
    \begin{equation}
        \hat{U}\ket{\bm\psi^{(n)}} = P\left( \sum_{i=1}^m U_{i1} \hat{b}_i^\dagger, \dots, \sum_{i=1}^m U_{im} \hat{b}_i^\dagger \right) \ket{\vec{0}}.
        \label{equ:poly_with_all_output}
    \end{equation}

    We are interested in the amplitude $d_n$ of the fully bunched state in the first output mode, denoted by $\ket{n, \vec{0}} = \frac{1}{\sqrt{n!}} (\hat{b}_1^\dagger)^n \ket{\vec{0}}$. To find this, we take the inner product of the target state with the evolved state from \cref{equ:poly_with_all_output}:
\begin{equation}
    d_n = \bra{n, \vec{0}} \hat{U} \ket{\bm\psi^{(n)}} = \bra{n, \vec{0}} P\left( \sum_{i=1}^m U_{i1} \hat{b}_i^\dagger, \dots, \sum_{i=1}^m U_{im} \hat{b}_i^\dagger \right) \ket{\vec{0}}.
\end{equation}
Because the bra $\bra{n, \vec{0}}$ is orthogonal to any Fock state containing bosons in modes $i > 1$, only the terms in the multinomial expansion of $P$ consisting exclusively of $\hat{b}_1^\dagger$ survive the inner product. This allows us to effectively set $\hat{b}_i^\dagger = 0$ for all $i > 1$ within the arguments of the polynomial:
\begin{equation}
    d_n = \bra{n, \vec{0}} P(U_{11}\hat{b}_1^\dagger, U_{12}\hat{b}_1^\dagger, \dots, U_{1m}\hat{b}_1^\dagger) \ket{\vec{0}}.
\end{equation}
Using the fact that $P$ is a homogeneous polynomial of degree $n$, we apply the scaling property $P(c z_1, \dots, c z_m) = c^n P(z_1, \dots, z_m)$ to pull the operator $\hat{b}_1^\dagger$ out of the polynomial:
\begin{equation}
    d_n = \bra{n, \vec{0}} (\hat{b}_1^\dagger)^n P(U_{11}, U_{12}, \dots, U_{1m}) \ket{\vec{0}}.
\end{equation}
Finally, noting that $(\hat{b}_1^\dagger)^n \ket{\vec{0}} = \sqrt{n!} \ket{n, \vec{0}}$, the expression simplifies to:
\begin{equation}
    d_n = \sqrt{n!} P(U_{11}, U_{12}, \dots, U_{1m}) \braket{n, \vec{0}|n, \vec{0}} = \sqrt{n!} P(U_{11}, \dots, U_{1m}).
    \label{equ:amplitude_dn}
\end{equation}

    Now, as a non-zero polynomial $P$ cannot vanish on all normalized vectors, there exists a vector $\vec{u}$ with $\|\vec{u}\|=1$ such that $P(\vec{u}) \neq 0$. By setting this $\vec{u}$ as the first row of $U$ (i.e.\ $U_{1j} = u_j$) and completing the matrix to be unitary (e.g.\ using the Gram--Schmidt process), we ensure $d_n \neq 0$.

    Finally, recall that the passive linear unitary $\hat{U}$ conserves the total boson number. Only the $n$-boson component $\ket{\bm\psi^{(n)}}$ can contribute to the $\ket{n}$ Fock state amplitude at the output. Thus, the reduced state $\ket{\psi_\mathrm{red}}= (\mathbb{I} \otimes \bra{0}^{\otimes m-1}) \hat{U} \ket{\bm\psi}$ has a non-zero $\ket{n}$ amplitude $d_n$.
\end{proof}

We note that to achieve the tightest bound possible with this procedure, the choice of unitary $U$ can be optimized for the specific state considered in order to maximize the bunching probability.

\subsection{Super-polynomial lower bound for \texorpdfstring{$\ket{1}^{\otimes n}$}{}}

\cref{thm:multimode_core_state_lowerbound_csr2} implies that the approximate coherent state rank of the $n$-mode Fock state $\ket1^{\otimes n}$ is lower bounded by $n+1$, i.e.\ a linear lower bound. Here we considerably strengthen this bound:

\begin{theorem}[Super-polynomial lower bound on the approximate coherent state rank of Fock states]\label{thm:superpoly}
Let $\ket{1}^{\otimes n}$ be the tensor product of $n$ Fock states $\ket{1}$. The approximate coherent state rank $\kappa(\ket{1}^{\otimes n})$ is $n^{\Omega(\log n)}$.
\end{theorem}

\noindent Note that this result is stated as asymptotic, but is also valid non-asymptotically, for $n$ sufficiently large.

\smallskip

Hereafter, given a polynomial $F$ over $n^2$ variables and an $n\times n$ matrix $X=(x_{i,j})_{1\le i,j\le n}$, we use the notation $F(X)$ to denote $F(x_{11},\dots,x_{1n},\dots,x_{n1},\dots,x_{nn})$.

\smallskip

The proof of \cref{thm:superpoly} proceeds as follows:
\begin{enumerate}
    \item We first relate the approximate coherent state rank $r=\kappa(\ket{1}^{\otimes n})$ to the size of multi-linear formulas approximating the permanent (\cref{lem:permanent_from_cs_superposition}). In particular, we show that there exists a sequence of formulas $F_\varepsilon$ over $n^2$ variables such that, for any $n\times n$ unitary matrix $U$, $|\mathrm{Per}(U)-F_\varepsilon(U)|\le\varepsilon$. Moreover, $F_\varepsilon$ is a multi-linear arithmetic formula (that is, the power of every input variable is at most one in each of its monomials) of size at most $rn^2$.
    \item Next we show that, for polynomials, such a uniform bound over the set of $n\times n$ unitary matrices implies convergence over the space of polynomials with degree at most $n^2$ (\cref{lem:unitary_convergence}). In particular, $F_\varepsilon$ has size at most $rn^2$ and converges to the permanent polynomial when $\varepsilon$ goes to $0$.
    \item Finally, we generalize the proof technique introduced by Raz in \cite{raz2009multi} for lower bounding the size of multi-linear formulas for the permanent to the approximate case (\cref{thm:bordercompper}). Setting $m=\lceil n^{1/3}\rceil$, we construct a partial-derivatives matrix (PDM) $M_A(F)$ of size $2^m$ from a given formula $F$ on $n^2$ variables by applying a random selection $A$ of its coefficients. On the one hand, the rank of $M_A(\mathrm{Per})$ is $2^m$ for all $A$, and the sequence $M_A(F_\varepsilon)$ converges to $M_A(\mathrm{Per})$, so the rank of $M_A(F_\varepsilon)$ is $2^m$ for $\varepsilon$ small enough. On the other hand, assuming $r=n^{O(\log n)}$, we show that there exists a subsequence along which the rank of $M_A(F_\varepsilon)$ is upper bounded by $rn^22^{m-k/2}$, where $k=n^{1/32}$, with probability $1-o(1)$. For $n$ sufficiently large, this leads to a contradiction.
\end{enumerate}

As outlined above, \cref{thm:superpoly} combines three intermediate results (\cref{lem:permanent_from_cs_superposition,lem:unitary_convergence,thm:bordercompper}). We first introduce and prove these results.

\begin{lemma}[Approximate multi-linear formulas for the permanent from coherent state decompositions]
\label{lem:permanent_from_cs_superposition}
Let $r=\kappa(\ket1^{\otimes n})$. For all $\varepsilon>0$, there exists a multi-linear arithmetic formula $F_\varepsilon$ of size at most $rn^2$ such that $|\mathrm{Per}(U)-F_\varepsilon(U)|\le\varepsilon$, for all unitary matrices $U$.
\end{lemma}

\begin{proof}
We have $\kappa(\ket1^{\otimes n})=r$ so for all $\delta>0$, $\kappa_\delta(\ket1^{\otimes n})\le r$. Let $\ket{\bm\phi_\delta} = \sum_{j=1}^r c_j(\delta)\ket{\bm\alpha_j(\delta)}$ be a superposition of $r$ coherent states such that $F(\bm\phi_\delta,\ket1^{\otimes n})\ge1-\delta$. We write $\bm\alpha_j(\delta)=(\alpha_{j1}(\delta),\dots,\alpha_{jn}(\delta))$.

Following \cite[Section 5.1]{Chabaud_2022_permanent}, for any $n\times n$ unitary matrix $U$ with associated $n$-mode passive linear unitary $\hat U$, we have:
\begin{align}
     \bra1^{\otimes n}\hat U\ket{\bm\phi_\delta}&= \sum_{j=1}^r c_j(\delta) \bra1^{\otimes n}\ket{U \bm{\alpha}_j(\delta)} \\
    &= \sum_{j=1}^r c_j(\delta) e^{-\frac{1}{2}\|U\bm{\alpha}_j(\delta)\|^2} \prod_{i=1}^n (U \bm{\alpha}_j(\delta))_i \\
    &= \sum_{j=1}^r c_j(\delta) e^{-\frac{1}{2}\|\bm{\alpha}_j(\delta)\|^2} \prod_{i=1}^n \left( \sum_{k=1}^n  \alpha_{jk}(\delta)\,u_{ik} \right),
\end{align}
where we used 
$\hat{U}\ket{\bm\alpha} = \ket{U\bm\alpha}$ in the first line, the inner product $\bra{1}\alpha\rangle = \alpha e^{-|\alpha|^2/2}$ in the second line, and
the fact that $\|U\bm{\alpha}\|^2 = \|\bm{\alpha}\|^2$ for unitary $U$ in the third line. Setting $\gamma_j(\delta)\coloneqq c_j(\delta) e^{-\frac{1}{2}\|\bm{\alpha}_j(\delta)\|^2}$, we obtain
\begin{equation}\label{eq:innermultilin}
    \bra1^{\otimes n}\hat U\ket{\bm\phi_\delta}= \sum_{j=1}^r \gamma_j(\delta) \prod_{i=1}^n \left( \sum_{k=1}^n \alpha_{jk}(\delta) \,u_{ik} \right).
\end{equation}

Up to adding a global phase, we may assume that $\bra1^{\otimes n}\ket{\bm\phi_\delta}$ is real and non-negative. With \cref{eq:linktoper} we thus obtain
\begin{equation}\label{eq:bounderrorper}
    \begin{aligned}
        |\mathrm{Per}(U)-\bra1^{\otimes n}\hat U\ket{\bm\phi_\delta}|&=|\bra1^{\otimes n}\hat U\ket1^{\otimes n}-\bra1^{\otimes n}\hat U\ket{\bm\phi_\delta}|\\
        &\le\|\ket1^{\otimes n}-\ket{\bm\phi_\delta}\|_2\\
        &\le\sqrt{2(1-F(\bm\phi_\delta,\ket1^{\otimes n}))}\\
        &\le\sqrt{2\delta}.
    \end{aligned}
\end{equation}

For any $n\times n$ matrix $X=(x_{ij})_{1\le i,j\le n}$ we now set 
\begin{equation}\label{eq:defFeps}
    F_\varepsilon(X)\coloneqq\sum_{j=1}^r\gamma_j(\varepsilon^2/2) \prod_{i=1}^n \left( \sum_{k=1}^n\alpha_{jk}(\varepsilon^2/2)\,x_{ik} \right).
\end{equation}
Since each product term involves only one variable from each row of $X$, the formula is multi-linear, and has size at most $rn^2$. Moreover, combining \cref{eq:innermultilin,eq:bounderrorper},
\begin{equation}
    |\mathrm{Per}(U)-F_\varepsilon(U)|\le\varepsilon,
\end{equation}
for all $n\times n$ unitary matrices $U$ and all $\varepsilon>0$.
\end{proof}

\begin{lemma}[Global convergence of multi-linear polynomials from unitary convergence]
\label{lem:unitary_convergence}
Let $\{P_m(X)\}$ be a sequence of complex multi-linear polynomials in the entries of an $n \times n$ matrix $X$ and let $P$ be a polynomial in $n^2$ variables of degree $d$. If $P_m(U)$ converges to $P(U)$ uniformly for all unitary matrices $U \in \mathcal U(n)$, then $P_m$ converges to $P$ over the space of complex polynomials in $n^2$ variables of degree at most $\max(n^2,d)$.
\end{lemma}

\begin{proof}
Let us first establish that if a complex polynomial $Q$ in $n^2$ variables vanishes identically on the unitary group $\mathcal U(n)$, it must be the zero polynomial. 

Consider the case of a univariate polynomial $q(z)$. A non-zero polynomial of degree $d$ has at most $d$ roots. Because the unit circle contains infinitely many points, $q(z)$ vanishing on the unit circle implies $q=0$. 

We can extend this result by induction on the number of variables $k$ of a polynomial $\bm z\mapsto Q(z_1, \dots, z_k)$ of degree $d$ which vanishes whenever all variables lie on the unit circle, i.e.\ $|z_i| = 1$ for all $1 \le i \le k$. Taking $w_2, \dots, w_k\in\mathbb C$ all on the unit circle, $z\mapsto Q(z,w_2, \dots, w_k)=\sum_{j=0}^dC_j(w_2, \dots, w_k)z^j$ is a univariate polynomial in $z$ which vanishes on the unit circle, and is therefore zero for all $z\in \mathbb{C}$. In particular, for all $0 \le j \le d$, the polynomial $C_j$ vanishes when $w_2, \dots, w_k$ are all on the unit circle. By the induction hypothesis, $C_j$ is the zero polynomial over $k-1$ variables, for all $0 \le j \le d$, and thus $Q$ must vanish for all $(z_1, \dots, z_k)\in\mathbb{C}^k$.

Now, consider the restriction of a polynomial $Q(X)$ in $n^2$ variables to diagonal matrices $Z = \mathrm{Diag}(z_1, \dots, z_n)$. If we require $|z_i| = 1$ for all $1 \le i \le n$, then $Z \in \mathcal U(n)$ and by hypothesis, $Q(Z)=Q(z_1,\dots,z_n) = 0$ (with a slight abuse of notation). By our previous observation, $Q$ must therefore vanish on all complex diagonal matrices. 

To generalize this to an arbitrary complex $n\times n$ matrix $X$, we use the singular value decomposition, $X = U \Sigma V^\dag$, where $U, V \in \mathcal U(n)$ and $\Sigma$ is a diagonal matrix. Define a new polynomial $R(Z) = Q(U Z V^\dag)$. For any unitary matrix $Z$, the product $U Z V^\dag$ is also unitary, meaning $R(Z) = 0$ for all diagonal $Z \in \mathcal U(n)$. Hence, $R$ must vanish on all diagonal matrices. Evaluating $R$ at the diagonal matrix $\Sigma$ yields $R(\Sigma) = Q(U \Sigma V^\dag) = Q(X) = 0$.

\smallskip

With this algebraic identity established, we let $Q_m(X) = P_m(X) - P(X)$, which converges uniformly to $0$ on $\mathcal U(n)$. Because the sequence $P_m(X)$ consists of multi-linear polynomials in the $n^2$ entries of $X$, the total degree of each $P_m$ is bounded by $n^2$, while the polynomial $P$ has degree $d$. Every $Q_m$ thus resides within the finite-dimensional vector space $V$ of complex polynomials in $n^2$ variables with a total degree bounded by $\max(n^2, d)$.
On this finite-dimensional space $V$, we define 
\begin{equation}
    \|\cdot\|_{\mathcal U}:Q\mapsto\|Q\|_{\mathcal U} = \sup_{U \in \mathcal U(n)} |Q(U)|.
\end{equation}
This function satisfies the standard properties of a norm: non-negativity, absolute homogeneity, and the triangle inequality are immediate, while definiteness ($\|Q\|_{\mathcal U} = 0 \implies Q = 0$) follows precisely from our earlier proof that any complex polynomial vanishing on $\mathcal U(n)$ is the zero polynomial. Hence, the sequence $P_m$ converges to $P$ on $V$.
\end{proof}

\begin{theorem}[Border complexity of multi-linear formulas for the permanent]\label{thm:bordercompper}
Any sequence of multi-linear arithmetic formulas $\{F_\varepsilon\}$ that converges to the $n \times n$ permanent must contain formulas of size $n^{\Omega(\log n)}$. As a result, the border complexity of multi-linear formulas for the permanent is super-polynomial.
\end{theorem}

\begin{proof} 
Assume, for the sake of contradiction, that there exists a sequence of multi-linear formulas $F_\varepsilon$ converging to the permanent coefficient-by-coefficient, and such that for all $\varepsilon$, the formula size is bounded as $|F_\varepsilon| \le n^{c \log n}$, for some constant $c$. By \cite[Proposition 2.1]{raz2009multi}, we can assume that the formulas are syntactic (products in the formulas only multiply variables from disjoint sets) without loss of generality.

A formula of size $S$ over $n^2$ variables is a binary tree with at most $S$ nodes, where internal nodes are labeled as $+$ or $\times$, and leaves are labeled with variables from $X = \{x_{i,j}\}$ or field constants. Because $S$ and $n$ are finite, the number of possible underlying tree structures and variable-to-leaf assignments is finite. By the pigeonhole principle, at least one specific (syntactic) structure, denoted $T^*$, must occur infinitely many times. We can therefore extract a subsequence (which we still denote as $\{F_\varepsilon\}$) such that every single formula $F_\varepsilon$ in the sequence possesses the same structure $T^*$, and the only variations as $\varepsilon \to 0$ are the field constants at the leaves\footnote{We include this extraction step for the generality of \cref{thm:bordercompper}, but it is not strictly necessary for the proof of \cref{thm:superpoly} since the formulas in \cref{eq:defFeps} already have the same structure for all $\varepsilon>0$.}.

Following \cite{raz2009multi}, we define a random assignment $A$ that substitutes the variables $X$ with values in $\{0, 1\} \cup \{y_1, \dots, y_m\} \cup \{z_1, \dots, z_m\}$ where $m = \lceil n^{1/3} \rceil$. Then, given a formula $F$, we denote by $M_A(F)$ the $2^m\times2^m$ partial-derivatives matrix of the (output of the) formula $F$ under the substitution $A$ (see \cite[Section 3]{raz2009multi} for a formal definition).

By \cite[Lemma 4.1 and Lemma 5.1]{raz2009multi}, for a syntactic formula $F$ of size bounded as $|F|\le n^{c \log n}$, the probability that $\mathrm{rank}(M_A(F))\le|F|\cdot2^{m-k/2}$
is $1-o(1)$, where $m = \lceil n^{1/3} \rceil$ and $k=n^{1/32}$. Importantly, this result depends exclusively on the structure of the formula $F$. Since the structure $T^*$ is fixed for our entire subsequence $\{F_\varepsilon\}$, this probability is independent of $\varepsilon$. Hence, assuming $|F_\varepsilon| \le n^{c \log n}$ we obtain
\begin{equation}
    \operatorname{rank}(M_{A}(F_\varepsilon)) \le |F_\varepsilon| \cdot 2^{m-k/2} \le n^{c \log n} \cdot 2^{m-k/2},
\end{equation}
for all $\varepsilon>0$, with probability $1-o(1)$.
We have $1-o(1)>0$ for large $n$, so for large $n$ there must exist at least one specific assignment $A^*$ such that $\operatorname{rank}(M_{A^*}(F_\varepsilon)) \le n^{c \log n} \cdot 2^{m-k/2}$ for all $\varepsilon>0$.
As $\varepsilon \to 0$, because the coefficients of $F_\varepsilon$ converge to those of the permanent, the entries of the partial-derivative matrix $M_{A^*}(F_\varepsilon)$ converge to the entries of $M_{A^*}(\mathrm{Per})$.
Since the rank of a matrix is lower semi-continuous we obtain:
\begin{equation}
    \operatorname{rank}(M_{A^*}(\mathrm{Per})) \le \liminf_{\varepsilon \to 0} \operatorname{rank}(M_{A^*}(F_\varepsilon)) \le n^{c \log n} \cdot 2^{m-k/2}.
\end{equation}
On the other hand, the partial-derivative matrix of the permanent has full rank \cite{raz2009multi}, namely
$\operatorname{rank}(M_{A^*}(\mathrm{Per})) = 2^m$.
We therefore arrive at the inequality $2^m \le n^{c \log n} \cdot 2^{m-k/2}$, or equivalently
\begin{equation}
    2^{k/2}\le n^{c \log n},
\end{equation} 
where $k = n^{1/32}$. For $n$ large enough, this leads to the desired contradiction. Therefore, no such bounded-size sequence of formulas can exist for $n$ large enough.
\end{proof}

We note that, akin to the result of Raz \cite{raz2009multi}, \cref{thm:bordercompper} is also valid for the determinant with the same proof.

\smallskip

With these results in place, we are now in a position to complete the proof of \cref{thm:superpoly}.

\begin{proof}[Proof of \cref{thm:superpoly}]
By \cref{lem:permanent_from_cs_superposition}, for all $\varepsilon>0$, there exists a multi-linear arithmetic formula $F_\varepsilon$ of size at most $n^2\kappa(\ket{1}^{\otimes n})$ such that $|\mathrm{Per}(U)-F_\varepsilon(U)|\le\varepsilon$, for all unitary matrices $U$.

\smallskip

By \cref{lem:unitary_convergence}, $F_\varepsilon$ converges to the permanent on the space of complex polynomials in $n^2$ variables of degree at most $n^2$. In particular, $F_\varepsilon$ converges to the permanent coefficient-by-coefficient and has arithmetic formula size at most $n^2\kappa(\ket{1}^{\otimes n})$, so the border complexity of multi-linear formulas for the permanent is upper bounded by $n^2\kappa(\ket{1}^{\otimes n})$.

\smallskip

By \cref{thm:bordercompper}, this implies that $n^2\kappa(\ket{1}^{\otimes n})=n^{\Omega(\log n)}$, and thus $\kappa(\ket{1}^{\otimes n})=n^{\Omega(\log n)}$.

\end{proof}

\bibliographystyle{myalphaurl}
\bibliography{bib}

\appendix
\crefalias{section}{appendix}
\crefalias{subsection}{appendix}

\section{Additional background}
\label{appendix:theory}

Here we provide additional background on continuous-variable (bosonic) quantum information theory.

\smallskip

The Hilbert space associated with a single continuous-variable quantum system, also known as a \emph{bosonic mode}, admits several unitarily equivalent realizations.

\smallskip

Let us consider the \emph{Fock space}, which is a mathematical object defined by $(\mathcal{H}, \hat{a}, \hat{a}^\dagger, \ket0 )$, where $\mathcal{H}$ is the Hilbert space, $\hat{a}$ and $\hat{a}^\dagger$ are the annihilation and creation operators and $\ket0$ is the vacuum state.

An abstract representation of the Fock space, spanned by the Fock states $\{\ket{n}\}_{n= 0}^\infty$, is defined by the action of the creation and annihilation operators $\hat{a}^\dagger$ and $\hat{a}$, such that $\hat{a}^\dagger\hat{a}\ket{n}=n\ket{n}$ and $\ket{n}=\frac{\left(\hat{a}^\dagger\right)^n}{\sqrt{n!}}\ket{0}$. The state $\ket{0}$ is the vacuum (Fock) state; it corresponds to the lowest-energy state of a quantized mode \cite{Sakurai_Napolitano_2020}. 

A famous realization in physics is the \emph{Schrödinger representation}. In this representation, states are square-integrable wavefunctions $\psi \in L^2(\mathbb{R})$ and the creation and annihilation operators are defined as 
\begin{equation}
    \hat{a} \mapsto \frac{1}{\sqrt2}\left(x+\frac{d}{dx}\right), \qquad \hat{a}^\dagger \mapsto \frac{1}{\sqrt2}\left(x-\frac{d}{dx}\right).
\end{equation} 
In this picture, the Fock states take the form
\begin{equation}
\psi_n(x) = \langle x | n \rangle
=
\frac{1}{\pi^{1/4} \sqrt{2^n n!}}
\, H_n(x)\, e^{-x^2/2},
\end{equation}
where $H_n$ denotes a Hermite polynomial \cite{Sakurai_Napolitano_2020}. 

A third realization is the \emph{Segal--Bargmann representation}, in which pure states are represented by holomorphic functions of a complex variable $z$ equipped with a Gaussian measure. This representation was introduced by Bargmann and Segal as a unitary map from $L^2(\mathbb{R})$ onto a Hilbert space of entire functions $L^2(\mathbb{C}, \mu_G)$ along with a Gaussian measure $\mu_G$ \cite{Bargmann_1961, segal_1967}. In this picture, the ladder operators act as
\begin{equation}
a^\dagger \mapsto z,
\qquad
a \mapsto \partial_z ,
\end{equation}
and a state $\ket{\psi} = \sum_{n=0}^\infty \psi_n\ket{n}$ has the following Segal--Bargmann representation (or stellar function \cite{Chabaud_2020}):
\begin{equation}
    F_\psi^\star(z) = \sum_{n=0}^\infty \psi_n \frac{z^n}{\sqrt{n!}}.
\end{equation} 
The stellar function is a holomorphic function over the
complex plane, providing an analytic representation of the quantum state, unique up to a global phase.

\begin{table}[ht!]
\renewcommand{\arraystretch}{1.35}
\begin{tabular}{p{0.24\textwidth} p{0.34\textwidth} p{0.34\textwidth}}
\toprule
Concept
&
Fock representation
&
Segal--Bargmann representation
\\
\midrule

\textit{Hilbert space realization}
&
$\mathcal H \cong \ell^2(\mathbb N_0)$
&
$\mathcal H \cong L^2(\mathbb C,\mu_{\mathrm G})$
\\

&
$\{|n\rangle\}_{n\ge0}$ 
&
$\mu_{\mathrm G}(z)=\pi^{-1}e^{-|z|^2}\,d^2z$,
\\
\\

\textit{Operator realization}
&
$a|n\rangle=\sqrt n\,|n-1\rangle$

&
$a \mapsto \partial_z$
\\
&
$a^\dagger|n\rangle=\sqrt{n+1}\,|n+1\rangle$
&
$a^\dagger \mapsto z$
\\
\\

\textit{Fock states}
&

$\ket{n}$
&
$
\frac{z^n}{\sqrt{n!}}
$
\\
\\

\textit{Core states}
&

$\ket{C_{r^\star}}$
&
$
P_{r^\star}(z)
$
\\
\\
\textit{Coherent states}
&
$
|\alpha\rangle
=
e^{-|\alpha|^2/2}
\sum_{n=0}^\infty
\frac{\alpha^n}{\sqrt{n!}}\,|n\rangle
$
&
$e^{\alpha z+ \gamma}$

\\
\\
\textit{Squeezed states}
&
$
|\xi\rangle
=
\frac{1}{\mathcal{N}}\sum_{n=0}^\infty(\lambda')^n\frac{\sqrt{(2n)!}}{2^nn!}\ket{2n}
$
&
$e^{\lambda z^2+ \gamma}$

\\
\\
\textit{Gaussian states}
&

$|\psi_{\mathrm G}\rangle
=
U_{\mathrm G}\,|0\rangle$
&

$\exp\!\left(
 a z^2 + b z + c
\right)
$
\\
\bottomrule
\end{tabular}
\caption{
Correspondence between the Fock representation and Segal--Bargmann
representations of the bosonic Fock space.
The Segal--Bargmann space realizes the Fock space as a
subspace of holomorphic functions in \(L^2(\mathbb C,\mu_{\mathrm G})\), where the Gaussian vacuum weight
is absorbed into the measure. $P_{r^\star}$ is a complex polynomial of degree $r^\star$, where $r^\star$ is the highest Fock number of the core state $\ket{C_{r^\star}}$. $U_G$ is a Gaussian unitary, generated by a Hamiltonian that is a polynomial of degree 2 in $\hat{a}$ and $\hat{a}^\dagger$.
}
\label{tab:fock-bargmann-correspondence}
\end{table}

Using the Segal--Bargmann representation, it becomes natural to introduce Fock states, coherent states, and Gaussian states. The stellar function of the $n^\text{th}$ Fock state $\ket{n}$ simply becomes
\begin{equation}
    F_{\ket{n}}^\star(z) = \frac{z^n}{\sqrt{n!}},
    \label{equ:fock_stellar_fct}
\end{equation}
while core states (finite superpositions of Fock states) have polynomial stellar functions.

Coherent states $\ket{\alpha}$ have the representation
\begin{equation}
    F_{\ket{\alpha}}^\star(z) = e^{\alpha z+\gamma},
    \label{equ:coherent_stellar_fct} 
\end{equation} where $\gamma$ normalizes the state. These states represent the ``most classical'' quantum states \cite{glauber_1963,sudarshan_1963}.

On the other hand, squeezed states $\ket{\xi}$ have the following representation:
\begin{equation}
    F_{\ket{\xi}}^\star(z) = e^{\lambda z^2+\gamma},
\end{equation} where $\lambda$ is a function of $\xi$ and $\gamma$ is a normalization factor.
Finally, we can define more general Gaussian states $\ket{G}$ as the states of the form
\begin{equation}
    F_{\ket{G}}^\star(z) = e^{a z^2+bz+c}, 
    \label{equ:gaussian_stellar_fct}
\end{equation} with $|a|<\frac{1}{2}$ and $c$ normalizing the state. This family of states is closed under the action of Gaussian unitary operations, which are defined precisely as those transformations that map any Gaussian state to another Gaussian state.

The correspondence between the Fock representation and the Segal--Bargmann representation is summarized in \cref{tab:fock-bargmann-correspondence}.

\section{Proof of \texorpdfstring{\cref{lemma:norm_H}}{}}
\label{appendix:proofs}

Here we give the proof of \cref{lemma:norm_H} used in the proof of \cref{thm:lower_bound_appox_coherent_stellar}. We recall the lemma below:
    
\setcounter{lemma}{0}

\begin{lemma}\label{lemma:norm_Happ}
Let $N\ge0$ and let $\ket\psi= \sum_{n=0}^\infty \psi_n\ket{n}$ and $\ket\phi= \sum_{n=0}^\infty \phi_n\ket{n}$ be two single-mode pure states, with associated Hankel matrices $H_N(\psi)=(\sqrt{(i+j)!}\,\psi_{i+j})_{0\le i,j\le N}$ and $H_N(\phi)=(\sqrt{(i+j)!}\,\phi_{i+j})_{0\le i,j\le N}$, respectively. Then, 
    \begin{equation}
        \|\ket\psi-\ket\phi\|_2^2\ge\frac{1}{(N+1)(2N)!}\|H_N(\psi)-H_N(\phi)\|_F^2.
    \end{equation}
\end{lemma}

\begin{proof}
We have $H_N(\psi)-H_N(\phi)=H_N(\psi-\phi)$, so by linearity it is enough to show $\|\ket\psi\|_2^2\ge\frac{1}{(N+1)(2N)!}\|H_N(\psi)\|_F^2$.
We have
\begin{align}
    \| H_N(\psi) \|_F^2 &= \sum_{i,j=0}^N (i+j)!|\psi_{i+j}|^2\\
    &\le(2N)!\sum_{i,j=0}^N|\psi_{i+j}|^2.
\end{align}   
For $n \in \{0, \dots, 2N\}$, the number of pairs $(i,j)$ such that $i + j = n$ and $0\leq i, j\leq N$ is given by
 \begin{equation}\label{eq:defmn}
m_n = 
\begin{cases}
n + 1 & 0 \leq n \leq N \\
2N - n +1 & N < n \leq 2N,
\end{cases}
 \end{equation} 
so $m_n \leq N + 1$ for all $n$.

Therefore,
 \begin{align}
    \| H_N(\psi) \|_F^2 &\le (2N)!\sum_{n=0}^{2N}m_n|\psi_n|^2\\
    &\le(N+1)(2N)!\sum_{n=0}^{2N}|\psi_n|^2\\
    &\le(N+1)(2N)!\|\ket\psi\|_2^2.
\end{align}  
\end{proof}

\section{Optimized single-mode lower bounds}
\label{app:optib}

In this section, we introduce a refined version of \cref{thm:lower_bound_appox_coherent_stellar} yielding a tighter family of bounds which may be optimized numerically, and investigate the tightness of these bounds on numerical examples.

\begin{theorem}[Optimized lower bounds on the $\varepsilon$-approximate coherent state rank via coefficient rescaling]
Let $\ket{\psi} = \sum_{n=0}^\infty \psi_n\ket{n}$ be a single-mode quantum state, let $N\in\mathbb N$ be a truncation parameter, let $b>0$ be a rescaling parameter, and let $H_{N,b}(\psi)$ be the Hankel matrix built from the first $2N+1$ rescaled coefficients $\{b^n\sqrt{n!}\psi_n\}_{n=0}^{2N}$ of $\ket{\psi}$, i.e.\ $(H_{N,b}(\psi))_{ij}=b^{i+j}\sqrt{(i+j)!}\,\psi_{i+j}$, for $0\leq i ,j \leq N$. Let
\begin{equation}\label{eq:defg}
    m_n \coloneqq \begin{cases}
        n+1 &0\leq n\leq N,\\
        2N-n+1 &N+1\leq n\leq 2N.
    \end{cases}
\end{equation}
Then, 
\begin{equation}
\label{eq:genlow2}
    \varepsilon < \max_{b>0,N\in\mathbb N}\left[\frac1{2\max_{n=0,\dots,2N}(m_nb^{2n}n!)}\sum_{l=r+1}^{N+1}\sigma_l(H_{N,b}(\psi))^2\right]\quad\Rightarrow\quad\kappa_\varepsilon(\psi)>r,
\end{equation}
for all $r\in\mathbb N$.
\label{thm:optimised_b}
\end{theorem}
    
\begin{proof}
    \label{proof:optimised_b}
The proof is a simple extension of that of \cref{thm:lower_bound_appox_coherent_stellar}.
Let
\begin{equation}
\ket{\phi} = \sum_{k=1}^r c_k \ket{\alpha_k}
\end{equation}
be an arbitrary superposition of $r$ distinct coherent states.
Since $b>0$ and $N\in\mathbb N$ are free parameters, taking the contrapositive in \cref{eq:genlow2}, we aim to show that the fidelity between the state $\ket\psi$ and any superposition of $r$ coherent states is upper bounded by
$1-\frac{1}{2\max_{n=0,\dots,2N}(m_nb^{2n}n!)}\sum_{l>r}\sigma_l(H_{N,b}(\psi))^2$, where $H_{N,b}(\psi)=(b^{i+j}\sqrt{(i+j)!}\,\psi_{i+j})_{0\le i,j\le N}$ is the Hankel matrix associated to the rescaled Fock amplitudes of the state $\ket\psi$. By \cref{eq:relnorm}, it is enough to show
\begin{equation}\label{eq:simplerbound2}
    \|\ket\psi-\ket\phi\|_2^2\ge\frac1{\max_{n=0,\dots,2N}(m_nb^{2n}n!)}\sum_{l=r+1}^{N+1}\sigma_l(H_{N,b}(\psi))^2.
\end{equation}
To prove this bound, we examine the representation of $\ket\phi$ in the Fock basis: we define
\begin{equation}
    \phi_n \coloneqq \braket{n|\phi} = \sum_{k=1}^r c_k e^{-|\alpha_k|^2/2} \frac{\alpha_k^n}{\sqrt{n!}}.
\end{equation}
Rescaling these coefficients by $b^n\sqrt{n!}$ yields 
\begin{equation}
    b^n\sqrt{n!}\,\phi_n=\sum_{k=1}^rd_k(b\alpha_k)^n,
\end{equation}
where we have set $d_k\coloneqq c_k e^{-|\alpha_k|^2/2}$. By \cref{eq:HVDV}, the associated Hankel matrix $H_{N,b}(\phi)=(b^{i+j}\sqrt{(i+j)!}\,\phi_{i+j})_{0\le i,j\le N}$ takes the form
\begin{equation}
    H_{N,b}(\phi)=V(b\bm\alpha)^TD(\bm d)V(b\bm\alpha),
\end{equation}
for all $N\ge r$, where $V(b\bm\alpha)$ is a $r\times(N+1)$ Vandermonde matrix and $D(\bm d)=\mathrm{Diag}(d_1,\dots,d_r)$. In particular, the rank of $H_{N,b}(\phi)$ is at most $r$, so by the Young--Eckart--Mirsky theorem (\cref{thm:YEM}),
\begin{equation}\label{eq:usingYEM2}
    \|H_{N,b}(\psi)-H_{N,b}(\phi)\|_F^2\ge\sum_{l=r+1}^{N+1}\sigma_l(H_{N,b}(\psi))^2,
\end{equation}
where $\|\cdot\|_F$ denotes the Frobenius norm.

Finally, we relate the Euclidean distance of states to the Frobenius distance of the associated Hankel matrices via the following technical result, whose proof is deferred to the end of the section:

\setcounter{lemma}{5}
\begin{lemma}
    Let $N\ge0$, let $b>0$, and let $\ket\psi= \sum_{n=0}^\infty \psi_n\ket{n}$ and $\ket\phi= \sum_{n=0}^\infty \phi_n\ket{n}$ be two single-mode pure states, with associated Hankel matrices $H_{N,b}(\psi)=(b^{i+j}\sqrt{(i+j)!}\,\psi_{i+j})_{0\le i,j\le N}$ and $H_{N,b}(\phi)=(b^{i+j}\sqrt{(i+j)!}\,\phi_{i+j})_{0\le i,j\le N}$, respectively. Then, 
    \begin{equation}
        \|\ket\psi-\ket\phi\|_2^2\ge\frac1{\max_{n=0,\dots,2N}(m_nb^{2n}n!)}\|H_{N,b}(\psi)-H_{N,b}(\phi)\|_F^2,
    \end{equation}
    where $m_n$ is defined in \cref{eq:defg}.
    \label{lemma:norm_H2}
\end{lemma} 

\noindent Combining \cref{lemma:norm_H2} with \cref{eq:usingYEM2} yields \cref{eq:simplerbound2}, and concludes the proof.
\end{proof}

We now give the proof of \cref{lemma:norm_H2}:

\begin{proof}
    Similar to the proof of \cref{lemma:norm_Happ} in \cref{appendix:proofs}, it is enough to show $\|\ket\psi\|_2^2\ge\frac1{\max_{n=0,\dots,2N}(m_nb^{2n}n!)}\|H_{N,b}(\psi)\|_F^2$ by linearity.
    We have
\begin{align}
    \| H_{N,b}(\psi) \|_F^2 &= \sum_{i,j=0}^N b^{2i+2j}(i+j)!|\psi_{i+j}|^2\\
    &=\sum_{n=0}^{2N}m_nb^{2n}n!|\psi_n|^2,
\end{align}   
where for $n \in \{0, \dots, 2N\}$, $m_n$ denotes the number of pairs $(i,j)$ such that $i + j = n$ and $0\leq i, j\leq N$ (see \cref{eq:defmn}).

Therefore,
 \begin{align}
    \| H_{N,b}(\psi) \|_F^2 &\le\max_{n=0,\dots,2N}(m_nb^{2n}n!)\sum_{n=0}^{2N}|\psi_n|^2\\
    &\le\max_{n=0,\dots,2N}(m_nb^{2n}n!)\,\|\ket\psi\|_2^2.
\end{align}  
\end{proof}

For Fock states $\ket N$, the Hankel matrix $H_{N,b}(\ket N)$ is anti-diagonal, with singular values equal to $b^N\sqrt{N!}$. Hence, taking $r=N$, \cref{thm:optimised_b} implies that $\kappa_\varepsilon(\ket N)=N+1$ whenever $\varepsilon<\frac{1}{2}\max_{b>0}\min_{n=0,\dots,2N}(\frac{b^{2N-2n}N!}{m_nn!})$, yielding a tighter bound than \cref{coro:coreFockACSR}.

\smallskip

We illustrate the tightness of the bounds from \cref{thm:lower_bound_appox_coherent_stellar} 
and its optimized version from \cref{thm:optimised_b} with two examples shown in 
\cref{fig:fock_bounds}. The left panel displays the infidelity between the superposition 
$\sqrt{1-\gamma}\ket{0}+\sqrt{\gamma}\ket{1}$, $\gamma\in[0,1]$, and the closest coherent state, alongside both bounds: this
shows how tight the bounds are at small Fock numbers. In particular, the optimized bound is closer to the exact value by one order of magnitude.
The right panel shows both bounds for
Fock states $\ket{n}$, $n=1,\dots,12$: the non-optimized bound decays super-exponentially 
due to the $(2N)!$ factor, dropping from $\mathcal{O}(10^{-1})$ at $n=1$ to 
$\mathcal{O}(10^{-17})$ at $n=12$, whereas the optimized bound mitigates this effect to remain at 
$\mathcal{O}(10^{-4})$ for $n=12$, showing significantly better scaling.

\begin{figure}[h!]
\centering
\includegraphics[width=0.48\textwidth]{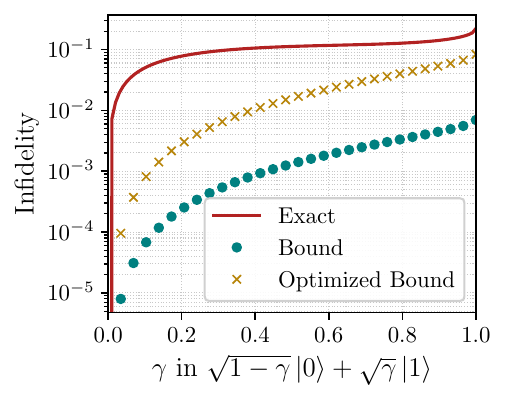}
\hfill
\includegraphics[width=0.48\textwidth]{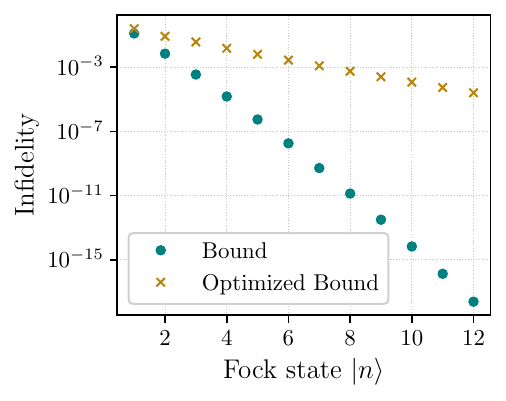}
\caption{%
    \label{fig:fock_bounds}
    Lower bounds on the infidelity between a target state $\ket\psi$ and the closest superposition of $\kappa(\psi)-1$ coherent states, showing the analytical bound from \cref{thm:lower_bound_appox_coherent_stellar} and the numerically optimized bound from \cref{thm:optimised_b}.
    \textit{Left:} Target state $\sqrt{1-\gamma}\ket{0}+\sqrt{\gamma}\ket{1}$, for which the exact value is also shown.
    \textit{Right:} Target state $\ket{n}$ for $n$ ranging from 1 to 12.
}
\end{figure}

\end{document}